\documentclass[conference]{IEEEtran}
\usepackage[dvipsnames]{xcolor}

\IEEEoverridecommandlockouts

\usepackage{tikz}
\usepackage{circuitikz}
\usetikzlibrary{automata,shapes.symbols,shapes.geometric,shadows,shapes.multipart,arrows.meta,positioning}
\usetikzlibrary{pgfplots.groupplots}
\usepackage{pgfplotstable}
\pgfplotsset{compat=1.17}

\usepackage{amsmath,amssymb,amsthm,mathtools}
\usepackage{pifont}
\newcommand{\cmark}{\ding{51}}%
\newcommand{\xmark}{\ding{55}}%

\usepackage{enumerate}

\usepackage{lstautogobble}  % Fix relative indenting
\usepackage{color}
\usepackage{zi4}            % Nice font
\usepackage[ruled,vlined]{algorithm2e}

\usepackage{etoolbox}
\expandafter\patchcmd\csname \string\lstinline\endcsname{%
  \leavevmode
  \bgroup
}{%
  \leavevmode
  \ifmmode\hbox\fi
  \bgroup
}{}{%
  \typeout{Patching of \string\lstinline\space failed!}%
}

\usepackage{listings}
\lstdefinelanguage{msp430asm}{
	morekeywords={mov,add,addc,sub,subc,cmp,dadd,bit,bic,bis,
				  xor,and,rrc,rra,push,swpb,call,reti,sxt,jeq,
				  jz,jne,jnz,jnc,jlo,jc,jhs,jn,jge,jl,jmp,
				  adc,dadc,dec,decd,inc,incd,sbc,inv,rla,rlc,
				  br,dint,eint,nop,ret,clr,clrc,clrn,clrz,pop,
				  setc,setn,setz,tst,INST1,INST2},
	sensitive=false,
	morecomment=[l]{;},
	morecomment=[s]{/*}{*/},
	morestring=[b]{"},
}
\lstdefinelanguage{alviespec}{
	keywords=[1]{mov,add,addc,sub,subc,cmp,dadd,bit,bic,bis,
				  xor,and,rrc,rra,push,swpb,call,reti,sxt,jeq,
				  jz,jne,jnz,jnc,jlo,jc,jhs,jn,jge,jl,jmp,
				  adc,dadc,dec,decd,inc,incd,sbc,inv,rla,rlc,
				  br,dint,eint,nop,ret,clr,clrc,clrn,clrz,pop,
				  setc,setn,setz,tst,INST1,INST2,zthen_else,timer_enable,create,jin,rst,rst_nz,start_counting,create_enclave,ubr,ifz,ifz_mov_nop,ifz_rst_nop,ifz_add_nop},
  keywords=[2]{cleanup,enclave,prepare,isr,att:,enc:},
	sensitive=false,
	morecomment=[s]{/*}{*/},
	morestring=[b]{"},
}
\makeatletter
\lstdefinestyle{mystyle}{
  basicstyle=%
    \ttfamily
    \lst@ifdisplaystyle\footnotesize\fi
}
\makeatother

\usepackage{graphicx}
\usepackage{verbatim}

\usepackage[numbers,compress]{natbib}

\usepackage{hyperref}
\usepackage[hyphenbreaks]{breakurl}
\usepackage[capitalise]{cleveref}

\newcommand{\Lstar}{\ensuremath{\rm L^*}\xspace}
\newcommand{\Lsharp}{\ensuremath{\rm L^\#}\xspace}
\newcommand{\fstar}{\ensuremath{\rm F^*}\xspace}

\usepackage{mystyle}

\usepackage{halloweenmath}
\usepackage{subcaption}

\usepackage{booktabs}
\usepackage{multirow}

\usepackage{placeins}
\usepackage{rotating}
\usepackage{flushend}

\begin{document}
%-------------------------------------------------------------------------------

\title{
  %\Large \bf
  % Bridging the Gap: Automated Analysis of Secure Embedded Architectures
  Bridging the Gap: Automated Analysis of Sancus
  \thanks{
    Work partially supported by projects
    ``SEcurity and RIghts In the CyberSpace - SERICS'' (PE00000014),
    ``Interconnected Nord-Est Innovation Ecosystem  - iNEST'' (ECS00000043), and
    ``Automatic Modelling and $\forall$erification of Dedicated sEcUrity deviceS - AM$\forall$DEUS" (PRIN P2022EPPHM),
     under the National Recovery and Resilience Plan (NRRP) funded by the European Union - NextGenerationEU.
  }
}

\author{
  \IEEEauthorblockN{Matteo Busi}
  \IEEEauthorblockA{
  \textit{Ca' Foscari University of Venice}\\
  Venice, Italy \\
  matteo.busi@unive.it}
\and
  \IEEEauthorblockN{Riccardo Focardi}
  \IEEEauthorblockA{
  \textit{Ca' Foscari University of Venice}\\
  Venice, Italy \\
  focardi@unive.it}
\and
  \IEEEauthorblockN{Flaminia Luccio}
  \IEEEauthorblockA{
  \textit{Ca' Foscari University of Venice}\\
  Venice, Italy \\
  luccio@unive.it}
}

% \author{Anonymized for submission}

\maketitle
\makeatletter
\def\ps@IEEEtitlepagestyle{
  \def\@oddfoot{\mycopyrightnotice}
  \def\@evenfoot{}
}
\def\mycopyrightnotice{
  {\footnotesize
  \begin{minipage}{\textwidth}
  % \centering
  To appear at IEEE CSF'24, July 8-12, 2024, Enschede, The Netherlands.
  \copyright~2024 IEEE.
  Personal use of this material is permitted.
  Permission from IEEE must be obtained for all other uses, in any current or future media, including reprinting/republishing this material for advertising or promotional purposes, creating new collective works, for resale or redistribution to servers or lists, or reuse of any copyrighted component of this work in other works.
  The definitive Version of Record is going to appear in Proceedings of the
  37th IEEE Computer Security Foundations Symposium (IEEE CSF'24), July 8-12, 2024, Enschede, The Netherlands.
  \end{minipage}
  }
}

%-------------------------------------------------------------------------------
\begin{abstract}
  Techniques for verifying or invalidating the security of computer systems have come a long way in recent years. Extremely sophisticated tools are available to specify and formally verify the behavior of a system and, at the same time, attack techniques have evolved to the point of questioning the possibility of obtaining adequate levels of security, especially in critical applications.
  In a recent paper,~\citet{mindthegap} have clearly highlighted this inconsistency between the two worlds: on one side, formal verification allows writing irrefutable proofs of the security of a system, on the other side concrete attacks make these proofs waver, exhibiting a gap between models and implementations which is very complex to bridge.
  In this paper, we propose a new method to reduce this gap in the Sancus  embedded security architecture, by
  exploiting some peculiarities of both approaches.
  Our technique first extracts a behavioral model by directly interacting with the real Sancus system and then analyzes it to identify attacks and anomalies.
  Given a threat model, our method either finds attacks in the given threat model or gives probabilistic guarantees on the security of the system.
  We implement our method and use it to systematically rediscover known attacks and uncover new ones.
\end{abstract}
%-------------------------------------------------------------------------------

%-------------------------------------------------------------------------------
\section{Introduction}\label{sec:intro}
%-------------------------------------------------------------------------------
% !TEX root = ../main.tex

Various approaches have been proposed in the literature for the analysis of the security of a computer system, ranging from formal specification and verification to the search for implementation vulnerabilities, also through the support of automated tools~\cite{ScienceSecurity,goguen1982security,patrignani2019formal,abate2018exploring,busi2021securing,NIprotocols,CVS,transient,SoKattacks,EmpiricalJavaScript,ResearchValueAttacks}. These approaches typically complement each other. Formal methods, for example, give the possibility to build implementations based on solid and verifiable specifications but, in most cases, the search for vulnerabilities remains necessary to identify errors, anomalies and, more generally, discrepancies with respect to the specification. In fact, even when security is irrefutably demonstrated through proofs carried out on formal models, there may be implementation-related bugs, which can be exploited by an attacker to subvert the security of the system. On the other hand, while the search for vulnerabilities provides evidence of bugs, alone is not sufficient to determine whether a system guarantees a security property. Indeed, the search cannot be exhaustive and, in the absence of a reference model, it is difficult if not impossible to clearly outline its perimeter and coverage.

In a recent paper,~\citet{mindthegap} have clearly discussed this gap between the formal, \emph{deductive} and the experimental, \emph{inductive} approach by showing practical examples of implementations that deviate from the formal specification, making the formal proofs waver.
In the literature, we find various approaches that try to reduce this gap, for example by generating code directly from a formal specification or, vice versa, by creating models as close as possible to the implementation.
The former approach uses sophisticated software engineering techniques based on formal methods to produce highly reliable code whose security guarantees are demonstrable~\cite{Koika,Kami,KamiCaseStudy,Fstar1,Fstar2}.
This approach is useful when developing new applications, but can hardly be applied to existing ones.
Instead, in the latter approach a model of an existing system is constructed using various techniques that may or may not require the active interaction with the system and the availability of the source code~\cite{reverseLego,reverseBank,UsenixTLSFuzzing,BLEfingerprint,vrased,verilog2SMV}.
This second approach, especially if applied \emph{black-box}, i.e., without knowing the source code of the system at hand, is particularly effective when analyzing existing systems, regardless of how they were developed.
However, these analyses usually target ad-hoc properties, and it is not immediate to relate them to formal results starting from deductive approaches.

In this work, we take a further step in bridging the gap between deductive and inductive approaches by proposing a method for the analysis of Sancus, a fully-functional prototype of a security architecture for networked embedded devices. Sancus extends the MSP430 processor, incorporating dedicated hardware support for memory access control and cryptographic operations. Notably, it introduces the concept of \emph{enclaved execution}, a protected memory region ensuring code and data encapsulation, effectively isolating them from the rest of the system. This allows the code to be executed in a protected environment so that both data confidentiality and integrity are, in principle, preserved. In~\citet{busi2021securing} Sancus was enriched with a formally verified mechanism for secure interruptible enclaves, but in a subsequent work~\citet{mindthegap} pointed out discrepancies between the formal model and its implementation, leading to a number of attacks and anomalies.

Our method for the analysis of Sancus is based on \emph{active automata learning}, a technique that allows to build a formal model of a real system by interacting with it in a black-box way~\cite{angluin1987learning,vaandrager2017model}. The peculiarity of our approach with respect to the state of the art, is its closer connection with the existing formal techniques. In fact, our aim  is to make it possible to formally verify
a general security property, typical of deductive approaches on real implementations. To achieve this ambitious goal, we pick ingredients from the two worlds so that the analysis has a value and is understandable from both perspectives: the general security property on which we focus is well understood and investigated in the formal, deductive world, but we instantiate it on concrete events so that the model extracted from the implementation and the anomalies highlighted correspond to concrete traces of execution, reproducible on the real system. This makes the analysis results valuable and understandable also in the experimental, inductive world.

We focus on \emph{robust noninterference}~\cite{abate2018exploring,busi2021securing}, a security property
requiring that the behavior of the sensitive part of a system, dealing with sensitive inputs and data, never triggers outputs or events that an active attacker can observe.
In particular, we consider a system with a trusted component $T$ that takes a sensitive input $\secret \in \secrets$, noted \T{}, and we let the attacker \A{} represent the untrusted context in which \T{} is executed, written \SA{\A}{T}{\secret}.
Robust noninterference holds if, for all secret inputs $s_0, s_1 \in \secrets$, the attacker \A{} cannot distinguish \Tl{} from \Tr{}, written:
\begin{equation} \label{prop:NI}
\SAl \oequiv \SAr
\end{equation}
where $\oequiv$ only observes the output and events that are visible to the attacker $A$.
% This formulation of robust noninterference is
% inspired by~\cite{busi2021securing}.

To verify this property on Sancus we have implemented a tool named \toolname{} that uses active automata learning~\cite{angluin1987learning,vaandrager2017model} to generate a finite state automaton which represents the behavior of the system, and then formally verifies \cref{prop:NI} on this automaton.
Learning certainly cannot be exhaustive and requires an abstraction that, in general, could be non-trivial to define.
In our case, this abstraction is driven by  \cref{prop:NI}: for example, since we are interested in checking for the absence of time-based side channels, we will actually observe the execution time.
\A{} and \T{} are specified as sequences of possible actions that are directly compiled into real system actions.
Specifying \A{} outlines in a clear and precise way the threat model and, consequently, the boundaries of the analysis.
Our method allows to specify \A{} and \T{} as regular expressions, providing good generality and computational tractability.
% This offers a good generality of the analysis and, at the same time, makes it computationally tractable.
Once we have learned the model of the system, we use an efficient technique to check \cref{prop:NI} and return, in case of failure, a representative counterexample for each interference allowing the attacker \A{} to distinguish \Tl{} from \Tr{}.
Counterexamples represent concrete and reproducible execution traces that can be run on the real system.

\citet{busi2021securing} use full abstraction to show that their implemented interruptible enclaves are secure, and that interrupts do not affect the attacker's ability to discriminate trusted components.
Intuitively, if we let $A'$ and $A$ denote the attacker with and without the capability of interrupting enclaved executions, it is proved that $T_0$ and $T_1$ are indistinguishable by $A$ if and only if they are indistinguishable by $A'$.
Full abstraction is very strong and in this particular scenario it implies the preservation of noninterference.
In fact, picking $T_0 = \Tl$ and $T_1 = \Tr$ we get:
\begin{equation} \label{prop:full}
% \forall s_0,s_1 \in \secrets,
% \SAl \oequiv \SAr \mbox{ iff } \SApl \oequiv \SApr
\SA{A}{T}{s_0} \oequiv \SA{A}{T}{s_1} \Rightarrow \SA{A'}{T}{s_0} \oequiv \SA{A'}{T}{s_1}
\end{equation}
This property is very appealing as it allows a developer to program \T{} without worrying about interrupts controlled by the attackers.
Here, we prove or disprove \cref{prop:full} by first learning the models of the four systems involved in the definition, and then by checking if the implication is satisfied.
If the equivalence in the premise of \cref{prop:full} holds but the one in the conclusion does not, we come up with counterexamples that break full abstraction, i.e., examples of trusted components \T{} that are secure with respect to $A$ ($\SAl \oequiv \SAr$), but are insecure with respect to $A'$ ($\SA{A'}{T}{s_0} \not\oequiv \SA{A'}{T}{s_1}$).
% We check \cref{prop:full} by learning the models of the four systems
% % \SAl{}, \SAr{}, \SApl{} and \SApr
% and then checking the two equivalences.
% % $\SAl \oequiv \SAr$ and $\SApl \oequiv \SApr$.
% When \cref{prop:NI} is satisfied but \cref{prop:full} is not, we actually come up with examples of trusted components \T{} that are secure with respect to $A$, i.e., $\SAl \oequiv \SAr$, but are insecure with respect to $A'$, i.e., $\SA{A'}{T}{s_0} \not\oequiv \SA{A'}{T}{s_1}$, breaking full abstraction.

Using this approach we were able to automatically discover new attacks on the latest Sancus commit~\cite{mindthegap-repo}, not previously documented in the literature, representing two new discrepancies between the formal model of~\citet{busi2021securing} and its implementation.
We were also able to systematically rediscover all the implementation-model mismatches presented by~\citet{mindthegap}
 except for one, due to some limitations in the implementation that we will discuss in~\cref{sec:knownattacks}.
The new attacks are counterexamples to \cref{prop:full} due to the possibility of resetting the machine from inside the trusted component either by an explicit reset, as well as through various access control violations. We found that resetting the trusted component can be leveraged by an attacker $A'$ to leak the control flow of the program, and consequently leak a sensitive branch guard, even when the reset happens in both branches.
Intuitively, the difference between $A$ and $A'$ is related to the possibility of measuring the time when the reset or the violation occurs thanks to the capability of interrupting the enclave.
In the Sancus formal model of~\citet{busi2021securing} resets and violations are handled explicitly making them observable both by $A$ and by $A'$.
In the Sancus implementation, instead, resets and violations simply reset the machine making $A'$ more powerful than $A$ and breaking full abstraction.

As a further confirmation of the correct functioning and scalability of the tool, we automatically analyzed the various Sancus commits, reconstructing the history of the fixes made after the results of~\cite{mindthegap}.
Interestingly, a developer could use \toolname to check that a new version of Sancus is working as expected and that it does not introduce regressions. In case no attack is found, the developer can conclude that Sancus preserves robust noninterference with respect to the specified attackers and trusted components with given confidence and precision.
This resembles more general results proved for secure compilation, where a program running under different computational models must exhibit the same behavior~\cite{patrignani2019formal,abate2018exploring}.

% \myskip
% \smallskip
\subsubsection*{Contributions}
We outline the most relevant contributions:
$(i)$ we propose a %general
    method to formally verify noninterference and its preservation on a real security architecture.
    As far as we know, our solution is the first one that verifies noninterference through black-box automata learning;
    % , a technique that creates a model of a real system by just interacting with it.
$(ii)$ we implement our technique in a tool called \toolname~\cite{additionalmaterial} and use it to analyze Sancus, rediscovering most known attacks and anomalies and finding new ones.
    \toolname is fully automated and can be run on various commits of a real system allowing for continuous formal verification,
    % .  \flaS{tolta altra frase e citazione. La parte dopo era il punto 5 messo assieme}
a practice increasingly requested by industry~\cite{facebookStatic};
  % It
  %is implemented as an \ocaml package, and is
  % available to the reviewers at~\cite{additionalmaterial}.
$(iii)$ we present the new attacks we found on the latest Sancus commit and discuss why they constitute further discrepancies between the formal model of \citet{busi2021securing} and Sancus implementation, complementing the results of \citet{mindthegap}. We discuss possible fixes, including new ones that diverge from the formal model and would increase, in our opinion, the security of Sancus;
$(iv)$ we show that, when an attack trace is discovered, \toolname can be used to synthesize the concrete attacker directly in the Sancus assembly language.
    This is useful in several respects: first, it provides a simple working version of the attack that can be reproduced on the real system to confirm the vulnerability, similar to a PoC provided by a security expert; furthermore, these synthesized attacks can be used as regression tests on other commits.
%5)  \toolname
    % we provide a docker image useful to reproduce our experiments.
    % We plan to make it open source, but have delayed its publication to preserve anonymity.
    % As a side contribution, we provide the first implementation in \ocaml of a recent active automata learning algorithm~\cite{vaandrager2022new} which, in its basic non-optimized form, has a performance comparable to previous highly optimized implementations of the classic Angluin's \Lstar{} algorithm~\cite{angluin1987learning}.
    % \flaS{tagliare anche qui?}

%-------------------------------------------------------------------------------
\section{Background and related work}\label{sec:related}
%-------------------------------------------------------------------------------
% !TEX root = ../main.tex
%In this section, we
We now recall the background and related work on model learning, formal analysis and various approaches that seek to bridge the gap between deductive and inductive methods.

\myskip
\subsubsection*{Model learning}
In recent years several techniques have been proposed for learning models from real hardware and software components \cite{vaandrager2017model}. The idea is to construct a model of a real System Under Learning (SUL) that can be further analyzed.
In passive model learning, the model is generated starting from execution traces that have been collected independently. In active model learning, instead, the interaction is guided by the learning process in order to improve the coverage and accuracy of the model. Model learning techniques can be white-box or black-box depending on whether or not they access the code, and differ according to the type of model generated. In a recent paper, \citet{vrased} have generated and verified a core model of an embedded device starting from the Verilog code using Verilog2SMV \cite{verilog2SMV}. This is a very interesting example of white-box model generation as the model is extracted from the code.
In this work, instead, we focus on black-box active automata learning, which adopts active model learning to build a model of a SUL in the form of finite-state automata, without any knowledge of the SUL code.
Black-box active automata learning is usually based on the classic \Lstar{} Angluin’s algorithm \cite{angluin1987learning} and its variants and optimizations, like those with faster counter-example processing~\cite{rivest1993inference}.
The \Lsharp{} algorithm~\cite{vaandrager2022new}
%was recently proposed which
%requires less implementation effort than \Lstar{}, but has comparable performance to previous highly optimized \Lstar{} implementations.
has comparable performance to previous highly optimized \Lstar{} implementations, but requires less implementation effort than \Lstar{}.

Black-box active automata learning has been successfully adopted, for example,  to reverse engineer Internet banking smartcard readers \cite{reverseLego}, bank cards \cite{reverseBank}, TLS implementations \cite{UsenixTLSFuzzing}, cache replacement policies \cite{CacheLearning}, for fingerprinting BLE devices \cite{BLEfingerprint}, and for analyzing TCP implementations \cite{TCPanalysis}.
All of these works learn a model of the SUL and then perform a specialized analysis on it.
Our approach is different as our focus is on the security property and the learning is driven by it.
Instead of aiming for reverse engineering or fingerprinting the SUL, we start from a noninterference property and instantiate it on a concrete threat model, so that what we learn from the SUL is what the attacker can observe.
In other words, guided by the security property, we specify the attacker's observation skills and capabilities
 producing a suitable abstraction which is used in the learning phase.
For what concerns automata learning, in this work we adopt the aforementioned  \Lsharp{} algorithm~\cite{vaandrager2022new}, providing
the first implementation  in \ocaml{}~\cite{additionalmaterial}.
\citet{chen2016pac} combine automata learning and concolic execution to verify the correctness of sequential programs.
Similarly to us, they provide probabilistic guarantees on correctness of their models.

\myskip
\subsubsection*{Secure compilation}
A prominent approach to narrowing the gap between specifications and implementations is to adopt particular toolchains that preserve properties at compile time, so that any formal guarantees proved on source code are preserved on the target machine.
This \emph{secure compilation} problem has been thoroughly investigated using formal methods (see, e.g., \cite{patrignani2019formal,abate2018exploring} for a survey).
In this work we consider the case study of Sancus, a security architecture that allows the execution of trusted code within isolated memory in a way that even a complete compromised system would not be able to tamper with the protected code and data.
A new mechanism for secure interruptible enclaves in Sancus has been introduced and formally verified by \citet{busi2021securing}, showing that a program executed in Sancus with or without interruptible enclaves behaves the same way, even when the attacker has full control of the device.
This technique, called full abstraction, provides strong formal guarantees but requires a model of the system in terms of formal semantics.
In a subsequent work, \citet{mindthegap} have shown a gap between the formal model of \cite{busi2021securing} and the actual implementation of Sancus. However, the analysis of \cite{mindthegap} was done by hand, looking for behaviors that could invalidate the assumptions behind the model of \cite{busi2021securing}.
In this work, we narrow this gap by presenting a technique and tool for systematically searching for these deviations that allowed us to find new attacks and rediscover most of issues raised in \citet{mindthegap}.

%Our technique can seamlessly integrate formal approaches based on secure compilation and full abstraction, reducing the gap between the formal models and the actual implementation.

\myskip
\subsubsection*{Verified compilation}
Another important approach in the literature aims at developing toolchains capable of producing verified code from a verified source. Interesting examples include Bluespec and its formal emanations Kôika and Kami \cite{Koika,Kami,KamiCaseStudy}, and the \fstar{} toolchains \cite{Fstar1,Fstar2}. These approaches are very interesting and effective, but require  that the system be designed and maintained using a particular language which is then compiled into machine code. Therefore, they cannot be applied to existing systems. Instead, our approach aims to analyze existing implementations in a black-box, code-independent way. This makes it suitable for verifying systems developed using mainstream, unverified languages and compilers.
Our results show that the analysis scales to a real prototype device and is efficient enough to allow for continuous formal verification on different commits.
This allows developers to check that a given fix is effective and, at the same time, to immediately spot regressions when new changes are committed.

\myskip
\subsubsection*{Noninterference}
In 1982, \citet{goguen1982security} introduced the idea of noninterference. Given two confidentiality security levels $\ell_1$ and $\ell_2$, we say that $\ell_1$ does not interfere with $\ell_2$ if any interaction performed at $\ell_1$ is not interfering with any observation done at $\ell_2$. For example if $\ell_1$ and $\ell_2$ respectively represent high and low confidentiality levels, noninterference ensures that sensitive, high level information cannot flow towards public, low level users.
In this work, we consider a formulation of noninterference, called \emph{robust noninterference}, in which the system has one trusted system component \T{} parametrized over a secret $\secret$ while everything else is controlled by an attacker \A{}, representing the untrusted execution environment in which \T{} is executed~\cite{abate2018exploring,busi2021securing}.

Given two secret values \secretl{} and \secretr{}, we require that the systems under attack \SAl{}, \SAr{} cannot be distinguished and, consequently, that there is no information leakage of the secret values towards the attacker. Since (robust) noninterference is extremely popular and well understood in deductive approaches, we use it as our reference property.

\myskip
\subsubsection*{Model validation and Fuzzing}
Other important techniques related to our work are model validation and fuzzing.
\citet{ValidationSide,ValidationSide2} have proposed techniques to validate side-channel observational models of computer architectures.
While we look for discrepancies between automatically extracted models of an actual processor under different attackers, they look for discrepancies between an actual processor and (user guided refinements of) its ISA level model.
%\matteoS{Ridetto come voleva il revisore. Ok?}
% In recent works, \citet{ValidationSide,ValidationSide2} have proposed techniques to validate side-channel observational models of computer architectures.
The automata we adopt here for model learning are very similar to the observation models of these papers but the aim is different, since they
validate existing models devised by hand with respect to the real system.
Here, we take the opposite approach by directly learning the model from the real system.
Fuzzing is a very popular technique for finding bugs and vulnerabilities on real systems and devices \cite{FuzzingBook}. In a recent work, \citet{FuzzingHardware} have proposed a technique  that uses standard software fuzzing tools  for hardware fuzzing. Our approach has some similarities with grammar-based black-box fuzzing as, e.g., \citet{FuzzingGrammar}, since we also use regular expressions to specify generic trusted and untrusted code that is subsequently used in the learning phase. However, our technique differs significantly from fuzzing in that it aims to generate a model of a system rather than collecting and analyzing crashes. This difference is fundamental because, as observed by \citet{clarkson2010hyperproperties}, noninterference-like properties are hyperproperties, i.e., properties of system behavior, that cannot be verified simply by looking for problematic traces that lead to crashes, as it is typically done for safety properties.

%-------------------------------------------------------------------------------
\section{Overview}\label{sec:overview}
%-------------------------------------------------------------------------------
% !TeX root = ../main.tex

In this section we introduce our verification methodology for detecting side-channel attacks or proving their absence in Sancus, with a given statistical confidence.
As far as we know, our proposal is the first that applies black-box automata learning, a technique that models real systems by interacting with them, in the context of side-channels attacks.

Sancus is a security architecture that provides enclaved execution, i.e., a protected memory in which trusted code and data are encapsulated and isolated from the rest on the system~\cite{noorman2017sancus}.
This allows the code to be executed in a protected environment so that both data confidentiality and integrity are, in principle, guaranteed (see~\cref{sec:related} for more detail).
For our analysis of Sancus, we consider two components: a set $A$ of attacks, and a set $T$ of trusted, enclaved programs which are parametrized over a secret $\secret \in \mb{S}$.
Let $\mi{Action}_A$ and $\mi{Action}_T$ denote the available actions for $A$ and $T$, respectively.
Let $\mc{O}$ represent the finite set of observable actions that the system can produce as a result of the aforementioned actions.
% Within this set, there exists a subset $\mc{O}_\tau$ referred to as \emph{silent}, intuitively modeling actions that an attacker cannot observe.
Let \SA{A}{T}{\secret} denote the execution of $T(\secret)$ under attack $A$ in Sancus represented by the set of execution traces over
% the  aforementioned actions
$\mi{Action}_A, \mi{Action}_T$, and $\mc{O}$. Let $\oequiv$ denote  a weak trace equivalence: $\SA{A}{T}{s} \oequiv \SA{A'}{T'}{s'}$ iff the execution traces of \SA{A}{T}{s} and \SA{A'}{T'}{s'} are the same up to some set of silent observables (see~\cref{sec:checking}).
At present, it is irrelevant whether or how these execution traces can be computed, or if the trace equivalence check is feasible or not. In \cref{sec:learning},
we will delve into a technique for automated learning of
\SA{A}{T}{s} behavior and its associated traces. This will be accomplished given specific
$A, T$ and $s$, and the resulting behavior will be represented using a finite-state Mealy machine.

A way of expressing that a system has no side-channels is to prove that any trusted component $t \in T$ is able to keep $\secret \in \mb{S}$ confidential under any attack in $A$:
\begin{definition}[RNI]\label{def:RNI}
    Sancus satisfies \emph{robust noninterference under T and A} ($\NI{A}{T}{\mb{S}}$) iff for any
    $s_0, s_1 \in \mb{S}$
    \[
        \SA{A}{T}{s_0} \oequiv \SA{A}{T}{s_1}.
    \]
\end{definition}
Depending on the instantiation of \emph{sets} $A$, $T$, and \mb{S} we may end up with very different notions of \NI{A}{T}{\mb{S}}.
For instance, one could consider a basic attacker
$\Abasic$
%, a singleton family comprising only one attack scenario noted $a^\flat$. In this scenario, the attacker initiates
that simply initiates the trusted component, observes its behavior, and waits for its termination.
\NI{\Abasic}{T}{\mb{S}} rejects trivially vulnerable programs that \emph{spontaneously} leak the secret $\secret$ without any active intervention from the attacker.
Inspired by~\citet{abate2018exploring}, we choose \NI{\Abasic}{T}{\mb{S}} as a baseline to distinguish interesting and potentially dangerous leakages from the others.
Let \Abasic be the basic attacker and \Aadvanced to be a richer, active attacker.
We define our new security notion as:
\begin{definition}[PRNI]\label{def:pres-NIA}
    Sancus satisfies the \emph{preservation of robust noninterference under \Abasic, \Aadvanced, and T} ($\PNI{\Abasic}{\Aadvanced}{T}{\mb{S}}$) iff for any
    %$t \in T$, $a^\flat \in \Abasic$, $a^\sharp \in \Aadvanced$ and
    $s_0, s_1 \in \mb{S}$
    \begin{equation*}\label{prop:fullinstance}
        \SA{\Abasic}{T}{s_0} \oequiv \SA{\Abasic}{T}{s_1} \Rightarrow \SA{\Aadvanced}{T}{s_0} \oequiv \SA{\Aadvanced}{T}{s_1}.
    \end{equation*}
\end{definition}
This property is very relevant in practical applications, as it allows developers of trusted components to work with just the capabilities of \Abasic in mind, without the need of worrying about what \Aadvanced might do.
Also, observe that different instantiations of \Abasic, \Aadvanced and $\oequiv$ allow us to specify different kind of side-channels and, thus, different threat models.
Furthermore, it is also important to notice that \Abasic, \Aadvanced, and T define sets of attacks or trusted programs, making PRNI a security notion for families of programs rather than individual ones.
%\matteoS{Chiarito così per rev che dice che PRNI nostra è "weak" (lo è, ma meno di quello che crede Rev. C :)}

Our analysis revolves around three phases: a \emph{specification phase}, a \emph{learning phase}, and a \emph{checking phase}.
In the first phase we define the security properties we are interested in, by specifying T, \Abasic, \Aadvanced, and $\oequiv$.
In the second phase we use the outputs of the first phase and an automata learning algorithm to learn a behavioral model of the trusted component under \Abasic and \Aadvanced.
Finally, in the third phase, we compare the learned models to detect any violation to the above preservation notion.

In order to explain the technique more in detail we will use a simple running example that,
beyond helping in the general understanding of the technique, will produce a clearly vulnerable program that is directly rejected by \NI{\Abasic}{T}{\mb{S}}.
%In order to explain the technique more in detail we will use a simple running example.
%Beyond helping in the general understanding of the technique, this example illustrates a clearly vulnerable program that is directly rejected by \NI{\Abasic}{T}{\mb{S}}.
In~\cref{sec:sancus} we will provide more significant cases of trusted programs that satisfy the \NI{\Abasic}{T}{\mb{S}} baseline but do not satisfy~\cref{def:pres-NIA}, corresponding to practical attacks to the Sancus architecture.

% also just looks at the basic attacker \Abasic and thus highlights the restrictiveness of \NIAbasic and the need of~\cref{def:pres-NIA}.
%
\medskip
\begin{example}[A simple timing attack]
\label{ex:timing1}
We consider a typical example of a timing attack, in which the branches of a comparison have a different execution time and, since the branch depends on the secret value \secret, the attacker can infer it, breaking confidentiality.
Only when \lstinline|s| is zero the value \lstinline|42| is written into another secret location \lstinline|&TMP|.
In both branches the program exits from the enclave by jumping to \lstinline|enclave_end|.
\begin{lstlisting}
enclave {
    cmp s, #0;       /* compare secret s with 0    */
    jnz enclave_end; /* if s != 0 exit enclave     */
    mov #42, &TMP;   /* only performed when s is 0 */
    jmp enclave_end  /* exit the enclave           */
}
\end{lstlisting}
%Notice that,
Since \lstinline|TMP| is in the protected memory, the attacker cannot directly observe whether the comparison was successful or not.
The trusted component $T(\secret)$ of~\cref{def:pres-NIA} corresponds to this enclaved code and data where we let $ s \in \mb{S} = \{ 0, 1 \}$.

Consider now a basic attack $\Abasic$ that runs the enclave and observes its execution time.
The result of the comparison is clearly leaked because the \lstinline|mov| command is executed only when \lstinline|s| is zero, taking strictly more time than the other branch, which does nothing.
\end{example}
\medskip
\noindent
\cref{ex:timing1} represents a buggy trusted component $T(\secret)$ which indirectly leaks \secret when executed in  Sancus.
% However, there is nothing in the Sancus enclaved execution that protects from this kind of problems, so this example does not represent an attack to Sancus, but just a buggy program: programmers in Sancus need to be aware that branches in $t$ should always be balanced in terms of execution time.
However, the Sancus enclaved execution does not protect against this type of issue. Therefore, this example does not depict an attack on Sancus but rather a flawed program. Programmers using Sancus need to be mindful that branches in
$T(\secret)$ depending on $\secret$ should always maintain a balance in terms of execution time.
% We use this example with the only purpose of illustrating how our technique can automatically uncover these timing attacks on a real implementation.
% , so it does not represent an attack , but we use it to illustrate our technique and, in particular, how it can automatically uncover this kind of
% timing attacks on a real implementation.
Our analysis can be applied to much more interesting cases where the program in the enclave is regarded as secure, but the underlying architecture has vulnerabilities, as we will discuss in detail in~\cref{sec:sancus}.

To improve readability, in the following we will use the abstract action \lstinline|ubr| that is translated into the unbalanced branch (lines 3,4,5 of the enclave of \cref{ex:timing1}), so that the above code can be conveniently rewritten as:
\begin{lstlisting}
enclave {
    cmp s, #0; /* compare secret s with 0          */
    ubr        /* unbalanced execution time branch */
}
\end{lstlisting}
%
%In the next sections we
We now analyze the execution of enclave of~\cref{ex:timing1}  in the latest version of Sancus\footnote{Commit \lastcommit in~\cite{,mindthegap,mindthegap-repo}.}, and we illustrate the three phases of our method: specification, learning and checking.

\subsection{Specification Phase}\label{sec:specification}
%The specification
This phase requires the precise definition of the families of  attacks and trusted components, and the observation power.

\myskip
\subsubsection*{Threat model}

This step amounts to precisely state the basic attacker \Abasic and the advanced attacker \Aadvanced.
These two components will constitute the threat model for the subsequent learning phase.
A Sancus attacker specification is composed of various sections that correspond to different attacker components, each of which is defined by a regular expression over the alphabet of actions $\mi{Actions}_A$.
A section may be declared inactive using the empty string identifier \lstinline+eps+, or may combine different actions using a regular expression-like syntax that supports choice `\lstinline+|+', sequencing `\lstinline+;+', and finite repetition `\lstinline+*+'.
Sections are of three different kinds: the \lstinline|isr| (interrupt service routine) section defines the code that may appear in the interrupt-service routine, executed when a timer interrupt is handled; the \lstinline|prepare| and \lstinline|cleanup| sections define the attacker behavior before and after the execution of an enclave, respectively.
Formally:
\begin{align*}\small
    \alpha &\in \mi{Actions} \qquad
    b \Coloneqq \lstinline+eps+ \mid \alpha \mid b\lstinline+;+ b \mid b \lstinline+|+ b \mid b\lstinline+*+\\
    A &\Coloneqq \lstinline+isr+\ \{ \,b\, \} \lstinline+; prepare+\ \{ \,b\, \} \lstinline+; cleanup+\ \{ \,b\, \}
\end{align*}
The definition of the language $\mc{L}(b)$ of a regular expression $b$ is standard, hence omitted.
Notice that, $\mc{L}(b)$ represents finite  sequences of $\mi{Actions}_A$ actions that can be thought as sequential programs. % \ricS{aggiunto questo}
Thus, the actual attacks associated with the specification $\lstinline+isr+\ \{ \,b_i\, \} \lstinline+; prepare+\ \{ \,b_p\, \} \lstinline+; cleanup+\ \{ \,b_c\, \}$ is the set of all programs of the form $\lstinline+isr+\ \{ \,\overline{\alpha}_i\, \} \lstinline+; prepare+\ \{ \,\overline{\alpha}_p\, \} \lstinline+; cleanup+\ \{ \,\overline{\alpha}_c\, \}$ with $\overline{\alpha}_i \in \mc{L}(b_i)$,  $\overline{\alpha}_p \in \mc{L}(b_p)$, and $\overline{\alpha}_i \in \mc{L}(b_c)$.
% and intuitively expresses all the attacks that may be used against a trusted component.
We abuse the notation and write $a \in A$ to denote that the attack $a$ is associated to $A$.

\begin{figure}[tb]
\begin{lstlisting}
isr { eps };             $\!$/* empty section        */

prepare {
    start_counting 256; /* start an 8-bit timer */
    create_enclave;     /* create the enclave   */
    jin enclave_start   /* start the enclave    */
};

cleanup  { eps }        /* empty section        */
\end{lstlisting}
\caption{\small The basic attacker $\Abasic$.}
\label{fig:basicattacker}
\end{figure}

\myskip
\begin{example}[Threat model, basic attacker \Abasic]
\label{ex:attacker}
Consider again our running example that  indirectly leaks the secret bit \lstinline|s| through a timing side-channel. Since the code is insecure, we consider a basic attacker that simply starts a timer, creates the enclave and jumps into it.
This is done in the \lstinline|prepare| section, and the other two sections are empty, as shown in~\cref{fig:basicattacker}.
Notice that, \lstinline|create_enclave| creates an enclave using the code specified for \T{} as, e.g., the one in~\cref{ex:timing1}.
It is a shortcut for \createenclave, which specifies the memory boundaries of the enclave code and data.
\end{example}

\myskip
\subsubsection*{Trusted component}
The trusted component specification is composed of just one section called \lstinline+enclave+ that describes the body of the trusted component and is defined by a regular expression over an alphabet of actions $\mi{Actions}_T$.
Similarly to attacker components, the section may be declared inactive using the empty string identifier \lstinline+eps+, or may combine different actions using a regular expression-like syntax:
\begin{align*}
    \alpha &\in \mi{Actions}_T \qquad
    b \Coloneqq \lstinline+eps+ \mid \alpha \mid b\lstinline+;+ b \mid b \lstinline+|+ b \mid b\lstinline+*+\\
    T &\Coloneqq \lstinline+enclave+\ \{ \,b\, \}
\end{align*}
The set of programs defined by $T$ is defined analogously to that of $A$ above, so we omit the definition here.
\begin{example}[Trusted component]
A simple example of a trusted component specification is~\cref{ex:timing1} that specifies a single trusted component, vulnerable to a timing attack.
\end{example}
\myskip
\subsubsection*{Abstract behavior}
In order to build a model of  executions it is necessary to define an \emph{abstract behavior} specification $\mc{B}$ of the system. % $S$.
Intuitively, $\mc{B}$ maps the concrete system outputs into the more abstract observables $\mc{O}$, capturing the security-relevant behavior of the system, for a particular analysis.

% \begin{example}[Observable behavior]
\myskip
\begin{example}[Abstract behavior]
\label{ex:abstract}
In order to analyze \cref{ex:timing1} we need to detect timing differences inside enclaves of Sancus, thus our abstract behavior maps the internal CPU clock and the CPU mode (i.e., whether it is executing an enclave or not) into the following abstract observables:
\begin{itemize}
    \item $\mi{JmpIn}$, denoting that the CPU transferred the control from the attacker to the enclave;
    \item $\mi{JmpOut}\ k$, denoting that the CPU took $k$ cycles to execute the given input and to transfer the control from the enclave back to the attacker;
    \item $\mi{Time}\ k$, denoting that the CPU took $k$ cycles to execute the input;
    % \item $\mi{Reset}$    denoting CPU reset;
    %and halt, respectively;\ricS{Ho commentato la $\tau$ queste servono ancora? Probabilmente solo la Reset?}\matteoS{yes}
    \item $\mi{TimerA\ k}$, a timer controlled and observed by the attacker used to detect timing side-channels.
    % \item $\tau$\ricS{Forse non c'è più?} denoting the fact that CPU performed other, uninteresting state transition.
    % Inputs causing such transitions include, for example, all the attackers' actions not listed above.
    % internal, silent state transition.
\end{itemize}
The complete abstract behavior of Sancus will extend the one above (cf.~\cref{sec:attackers_enclaves}).
\end{example}

\subsection{Learning Phase}
\label{sec:learning}

% Once all the specifications are in order, we are ready to start the learning phase.
The goal of this phase, is to build a deterministic finite-state Mealy machine \modelA{A}{T}{s} which describes the behavior of \SA{A}{T}{s}.
%, i.e., the behavior of \SAi for a give secret \secreti. % (cf. \cref{def:RNI}).
% $\mu^\mc{O}_{\mc{A}, \mc{V}, C}$ that describes how the victims specified by $\mc{V}$ behave on $C$ when attacked as described by $\mc{A}$.
Performing the learning phase for \Abasic, \Aadvanced, T and for any $s \in \mb{S}$ it is possible to formally verify that the specified system satisfies~\cref{def:pres-NIA}, as we will see in \cref{sec:checking}.
% Repeating the learning phase once for every choice of attacker specification \mc{A} and, possibly, SUL versions gives us a collection of formal models to be checked and compared in the next phase of the technique.
% Repeating the learning phase once for every choice of attacker specification \mc{A}, victims' specification \mc{V}, and CPU version $C$ gives us a collection of formal model to be checked and compared in the next phase of the technique.
%
\tikzset{
    every node/.style={
        on grid=false,
        node distance=1.5cm,
        align=center,
        thick},
    specs/.style={
        shape=tape,
        draw,
        fill=white,
        thick,
        tape bend top=none,
        double copy shadow},
    spec/.style={
        shape=tape,
        draw,
        fill=white,
        thick,
        tape bend top=none,
        align=center},
    predproc/.style={
        rectangle split,
        rectangle split parts=3,
        draw,
        thick,
        rectangle split horizontal=true,
        rectangle split ignore empty parts=false,
        rectangle split empty part width=0
    },
    io/.style={
        shape=trapezium,
        draw,
        thick,
        trapezium left angle = 65,
        trapezium right angle = 115,
        trapezium stretches
    }
}
Our learning phase is an instance of the active automata learning framework proposed in~\cite{angluin1987learning}.
The framework assumes a learning algorithm, the \emph{learner}, that interacts with a \emph{teacher} acting as an intermediary between the learner and the system that implements a hidden Mealy automaton.
The learner tries to infer the hidden automaton by repeatedly asking two types of queries to the teacher:
\begin{enumerate}[(i)]
    \item \emph{Membership queries:} the learner asks whether a sequence of actions in $(\mi{Action}_T \uplus \mi{Action}_A)^*$ belongs to those accepted by the hidden automaton;
    \item \emph{Equivalence queries:} the learner asks the teacher if an automaton accepts the same language as the hidden one.
\end{enumerate}
This phase is rendered as an \ocaml program (cf. \cref{sec:alvie}) that implements a modified version of the recently-proposed \Lsharp{} algorithm~\cite{vaandrager2022new}, and interacts with the real system under attack, i.e., \SA{A}{T}{s}, to build the required model.
In the ideal case, that is when the equivalence check is implemented using an oracle, the following correctness theorem holds:
% for the \Lsharp algorithm:
%
% \medskip
\begin{theorem}[\Lsharp correctness~\cite{vaandrager2022new}]\label{thm:ideal-correctness}
    Assume that the behavior of the system under specifications $A$ and $T$ is expressible as a Mealy machine $m$.
    The algorithm \Lsharp terminates and returns a model \modelA{A}{T}{s} that recognizes the same language as $m$.
\end{theorem}
% \medskip
% \noindent
%
Answering to equivalence queries, i.e., establishing an equivalence between the candidate model and the behavior of a system under A and T, is undecidable in general.
For this reason, various approximate approaches have been proposed in the literature~\cite{angluin1987learning,muskardin2022aalpy}.
We resort to a \emph{Probably Approximately Correct (PAC)}~\cite{angluin1987learning,muskardin2022aalpy} oracle, which looks for a counterexample to the equivalence between the model and the system by systematically sampling an increasing number of input queries.
More precisely, it takes as parameters $\epsilon > 0$ and $\delta > 0$ and at its $r$-th invocation it tries to distinguish the model and the system using $\lceil \frac{1}{\epsilon}(\ln{\frac{1}{\delta}} + r\ln{2}) \rceil$ randomly selected input sequences.
We call $(\delta, \epsilon)$-\Lsharp the learning algorithm modified to use such an approximated equivalence oracle that provides precise probabilistic bounds on the correctness of the model.
Let $\mc{D}$ be the probability distribution on the set of possible action sequences according to $A$ and $T$.
We say that $\mu$ is an $\epsilon$-approximation of $m$ iff the sum of the probabilities according to $\mc{D}$ of sequences accepted by $\mu$ and not by $m$ (or vice versa) is at most $\epsilon$, i.e., the probability of an error in the model is at most $\epsilon$.
% \matteoS{Aggiunto questo per Rev. B. Boh?}
%
% \medskip
\begin{theorem}[PAC learning correctness]\label{thm:pac-pni}
Assume that the behavior of the system under specifications $A$ and $T$ is expressible as a Mealy machine $m$,
% that $\mc{D}$ is the probability distribution on the set of possible action sequences according to $A$ and $T$,
% $\mc{D}$ as above,
and that the sampling in the PAC oracle is performed according to $\mc{D}$.
Then, $(\delta, \epsilon)$-\Lsharp terminates and with probability at least $1-\delta$ it returns a model \modelA{A}{T}{s} that is an $\epsilon$-approximation of $m$.
%, i.e., the probability of an error in the model is at most $\epsilon$.
\end{theorem}
% \medskip
%
\begin{proof}[Proof (Sketch)]
Termination follows by the termination of \Lsharp~\cite{vaandrager2022new}, which does not make any assumption on the way the equivalence oracle is implemented.
Instead, the probabilistic guarantees are easily recovered from the original results by~\citet{angluin1987learning} on automata learning by sampling.
\end{proof}
In order to sample input sequences to be used during the execution of $(\delta, \epsilon)$-\Lsharp, we devised a simple sampling algorithm that incrementally builds an actual attack $a \in A$ and a trusted component $t \in T(\secret)$.
Starting from the empty $a$ and $t$ and an empty execution trace $\beta$ we repeatedly extend $\beta$ with a new admitted action $b$ (if any), and accordingly extend $a$ and $t$ when necessary.
More precisely, we first select $b$ uniformly at random from the set of actions accepted by the specification after executing $\beta$. %, written \executes{\beta.b}.
To decide if $b$ is accepted by the specification, we assume to know which section of the specification is currently being executed by Sancus, and use this information to check the admissibility of $b$ on the appropriate regular expression.
Then, we check whether the new execution trace $\beta b$ is already in $a$ or $t$.
If that is the case, we maintain $a$ and $t$, otherwise we extend $a$ or $t$ by appending $b$ to the actions of the currently-executing section.
Each extension of $\beta$ is done with a certain probability, so we do not impose any fixed constraint on the length of sequences and the expected running time of the generation is finite.
This simple process induces a probability distribution on the language accepted by the specifications $A$ and $T(s)$ as required by the formal results of \cref{sec:checking}.
%
% In order to check \executes{\beta.b} and \executesl{\beta.b} we \ricS{TODO, spiegazione intuitiva del passaggio da una sezione all'altra?}
%
% \matteo{Non credo di capire. Non mi è chiaro dove si veda il passaggio tra sezioni e non capisco bene il significato di \executesl{\cdot}. }
%
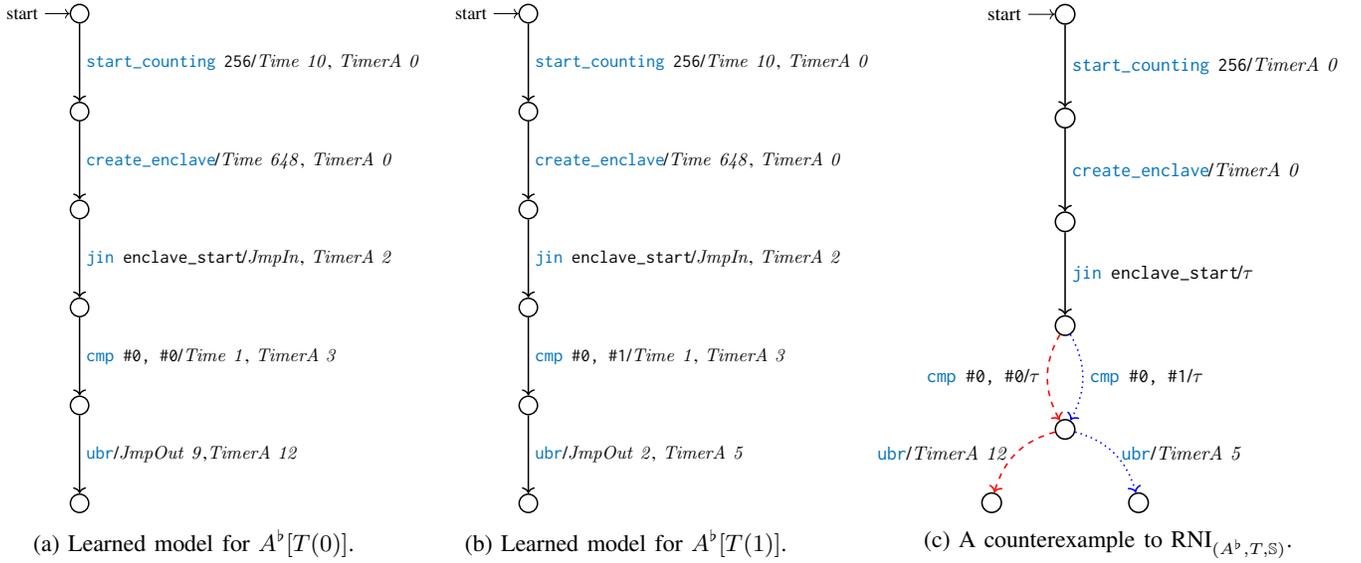
\begin{figure*}
    \centering
     \begin{subfigure}[t]{0.6\columnwidth}
     \centering
        \resizebox{!}{6.8cm}{
        \begin{tikzpicture}
            \node (s0) [state, initial,inner sep=0pt, minimum size=1em] {};
            \node (s1) [state, below = of s0,inner sep=0pt, minimum size=1em] {};
            \node (s2) [state, below = of s1,inner sep=0pt, minimum size=1em] {};
            \node (s3) [state, below = of s2,inner sep=0pt, minimum size=1em] {};
            \node (s4) [state, below = of s3,inner sep=0pt, minimum size=1em] {};
            \node (s5) [state, below = of s4,inner sep=0pt, minimum size=1em] {};

            \path [->, thick] (s0) edge node[right] {\lstinline{start\_counting 256}/\mi{Time\ 10}, \mi{TimerA\ 0} } (s1);
            \path [->, thick] (s1) edge node[right] {\lstinline{create\_enclave}/\mi{Time\ 648}, \mi{TimerA\ 0} } (s2);
            \path [->, thick] (s2) edge node[right] {\lstinline{jin enclave\_start}/\mi{JmpIn}, \mi{TimerA\ 2}} (s3);
            \path [->, thick] (s3) edge node[right] {\lstinline{cmp \#0, \#0}/\mi{Time\ 1}, \mi{TimerA\ 3}} (s4);
            \path [->, thick] (s4) edge node[right] {\lstinline|ubr|/\mi{JmpOut\ 9},\mi{TimerA\ 12}} (s5);
        \end{tikzpicture}
        }
        \caption{\small Learned model for \SA{\Abasic}{T}{0}.}
        \label{fig:learning1}
     \end{subfigure}
     \hfill
     \begin{subfigure}[t]{0.55\columnwidth}
        \centering
           \resizebox{!}{6.8cm}{
           \begin{tikzpicture}
            \node (s0) [state, initial,inner sep=0pt, minimum size=1em] {};
            \node (s1) [state, below = of s0,inner sep=0pt, minimum size=1em] {};
            \node (s2) [state, below = of s1,inner sep=0pt, minimum size=1em] {};
            \node (s3) [state, below = of s2,inner sep=0pt, minimum size=1em] {};
            \node (s4) [state, below = of s3,inner sep=0pt, minimum size=1em] {};
            \node (s5) [state, below = of s4,inner sep=0pt, minimum size=1em] {};

            \path [->, thick] (s0) edge node[right] {\lstinline{start\_counting 256}/\mi{Time\ 10}, \mi{TimerA\ 0} } (s1);
            \path [->, thick] (s1) edge node[right] {\lstinline{create\_enclave}/\mi{Time\ 648}, \mi{TimerA\ 0}} (s2);
            \path [->, thick] (s2) edge node[right] {\lstinline|jin enclave\_start|/\mi{JmpIn}, \mi{TimerA\ 2}} (s3);
            \path [->, thick] (s3) edge node[right] {\lstinline{cmp \#0, \#1}/\mi{Time\ 1}, \mi{TimerA\ 3}} (s4);
            \path [->, thick] (s4) edge node[right] {\lstinline{ubr}/\mi{JmpOut\ 2}, \mi{TimerA\ 5}} (s5);
           \end{tikzpicture}
           }
        \caption{\small Learned model for \SA{\Abasic}{T}{1}.}
                \label{fig:learning2}
     \end{subfigure}
     \hfill
    \begin{subfigure}[t]{0.75\columnwidth}
        \centering
        \resizebox{!}{6.8cm}{\
            \begin{tikzpicture}
                \node (s0b) [state, initial, inner sep=0pt, minimum size=1em] {};
                \node (s0) [state, below = of s0b, inner sep=0pt, minimum size=1em] {};
                \node (s1) [state, below = of s0, inner sep=0pt, minimum size=1em]       {};
                \node (s2) [state, below = of s1, inner sep=0pt, minimum size=1em]       {};
                \node (s3) [state, below = of s2, inner sep=0pt, minimum size=1em]       {};
                \node (s4) [state, below left = of s3, inner sep=0pt, minimum size=1em]       {};
                \node (s5) [state, below right = of s3, inner sep=0pt, minimum size=1em]       {};

                % Common part
                \path [->, thick] (s0b) edge node[right] {\lstinline{start\_counting 256}/\mi{TimerA\ 0}} (s0);
                \path [->, thick] (s0) edge node[right] {\lstinline{create\_enclave}/\mi{TimerA\ 0}} (s1);
                \path [->, thick] (s1) edge node[right] {\lstinline{jin enclave\_start}/$\tau$} (s2);
                % Interference
                \path [->, thick, dashed, bend right, color=red, text=black] (s2) edge node[left] {\lstinline{cmp \#0, \#0}/$\tau$} (s3);
                \path [->, thick, dotted, bend left, color=blue, text=black] (s2) edge node[right] {\lstinline{cmp \#0, \#1}/$\tau$} (s3);
                \path [->, thick, dashed, bend right, color=red, text=black] (s3) edge node[left] {\lstinline{ubr}/\mi{TimerA\ 12}} (s4);
                \path [->, thick, dotted, bend left, color=blue, text=black] (s3) edge node[right] {\lstinline{ubr}/\mi{TimerA\ 5}} (s5);
            \end{tikzpicture}
        }

        \caption{\small A counterexample to \NI{\Abasic}{T}{\mb{S}}.
        %: the attacker observes differences in \mi{TimerA}'s values and deduces the secret value in the comparison which is not observable directly.
         % given that the abstract output of \lstinline|cmp| is $\tau$.
         }\label{fig:tree-toy}
    \end{subfigure}

     \caption{\small The results of the learning and checking phases of~\cref{ex:learning,ex:checking}.}\label{fig:learned-toy}
\end{figure*}

\myskip
\begin{example}[Model learning]
\label{ex:learning}
We ran the learning phase on our running example using the specification $T$ of~\cref{ex:timing1}, the attacker \Abasic of \cref{ex:attacker}, and the abstract behavior of~\cref{ex:abstract}, obtaining the two models reported in
\cref{fig:learning1,fig:learning2}.
The transitions are labelled with the input action and the  corresponding output abstract actions separated by the `/' symbol.
E.g., \lstinline{jin enclave_start}/\mi{JmpIn},\mi{TimerA\ 2} is the start of the enclave which is observed as the abstract action \mi{JmpIn}.
The attacker timer \mi{TimerA} has value $2$.
The only difference is in the unbalanced branch input \lstinline|ubr| since the execution time is uneven in the two branches.
In fact, the \lstinline|mov| command is executed only when \lstinline|s| is zero (cf.~\cref{ex:timing1}).
This is respectively  observable through the \mi{JmpOut\ 9},\mi{TimerA\ 12} and \mi{JmpOut\ 2},\mi{TimerA\ 5} abstract outputs reporting that the execution of the unbalanced branch takes 9 clock cycles in \modelA{\Abasic}{T}{0} and 2 clock cycles in \modelA{\Abasic}{T}{1}.
The 7 clock cycles of difference are due to the extra \lstinline|mov| and \lstinline|jmp| that take, respectively, 5 and 2 clock cycles.
The attacker can observe this difference using \mi{TimerA} as explained in the next section.
% This is observable by the attacker as we will discuss in the
% We point out some interesting facts: in state 2, the two attackers use two different guesses, $0$ and $1$, which are stored in \lstinline|r5|, the register that is compared with the secret enclave bit; in state 4, the unbalanced branch in the enclave (\lstinline|ubr|) takes a different time to execute in the two models: 4 and 11, respectively; in states 5 and 6, \lstinline|TAR| is even in \modelAl{} while it is odd in \modelAr{}, triggering two clearly observable behaviors: \mi{Halt} vs. \mi{Reset}.
% Both models depicts a rather straightforward behavior, but show the important fact that the timing of the enclave, i.e., the observables associated with \lstinline|ubr|,\ricS{Cos'è ubr?} depends on the choice of \lstinline|r5| by the attacker and that the attacker can detect such a difference by either halting or resetting the CPU, in dependence of the elapsed time as measured by \lstinline|TAR|.

% \fla{dire qualcosa di JmpOut che prima non c'era e dei valori associati}
% \matteoS{Ho abbozzato due righe di descrizione.}
% \ricS{Secondo me vanno descritti un po' di più. Deve essere tutto chiaro a questo punto. Ad esempio è delicato il fatto che ci siano delle differenze prima del reset che sono osservabili ma di fatto immagino vengano nascoste dopo ... qui non è chiaro.}

\end{example}

% \begin{figure}[tb]
%     % \resizebox{\columnwidth}{!}{
%     \resizebox{\columnwidth}{!}{
%     \begin{tikzpicture}
%         %%% Nodes
%         \node[predproc,minimum height=1.5cm,minimum width=2.43cm] (cp)    {\nodepart{two}mCRL2~\cite{mcrl2}\\ detected violation?};
%         \node[io] (m1) [above left=of cp] {Learned model\\\modelAl{}};
%         \node[io] (m2) [above right=of cp] {Learned model\\\modelAr{}};
%         \node[io] (ok) [below right=of cp] {No violations to NI};
%         \node[io] (cex) [below left=of cp] {Witness graph $\mc{G}$\\ of interfering behavior};
%         \node[spec] (o_spec) [above=2cm of cp] {Observable\\ behavior\\ spec.\\$\mc{O}$};

%         %%% Arrows
%         \draw[-latex,thick] (m1) -| (cp.140);
%         \draw[-latex,thick] (m2) -| (cp.40);
%         \draw[-latex,thick] (o_spec) -- (cp);
%         \draw[-latex,thick] (cp) -| (ok)  node [pos=0.75,fill=white]{No};
%         \draw[-latex,thick] (cp) -| (cex) node [pos=0.75,fill=white]{Yes};
%     \end{tikzpicture}
%     }

%     \caption{
%         The checking phase.
%         % on two learned models $\mc{\mu}, \mc{\mu'}$ and an observable behavior $\mc{O}$.
%     }\label{fig:checking-phase}
% \end{figure}

\subsection{Checking Phase}
\label{sec:checking}

% The checking phase is the last of our technique and
During this phase we check whether the learned models adhere to~\cref{def:pres-NIA}, catching violations to the preservation of \PNI{\Abasic}{\Aadvanced}{T}{\mb{S}}.
%
% full-abstraction \cite{busi2021securing}.\ricS{Ho aggiunto questa frase per spiegare che qui mostriamo la funzionalità di base OK?}
% In practice, the depicted process is repeated once for any pair of learned models.
%
%
% Apart from converting learned models from our internal format to a suitable format for mCRL2, the only non-trivial task in implementing this phase is how to configure the model checker to correctly recognize interfering behavior.
% Correctly recognizing violations to noninterference using mCRL2 requires an extra abstraction.
Let $\mc{O}_\tau \subset \mc{O}$ to be a set of observables that the attacker cannot directly observe.
Ignoring silent observables, we check for any violation of \PNI{\Abasic}{\Aadvanced}{T}{\mb{S}}, i.e., we check for violations to \NI{\Aadvanced}{T}{\mb{S}} in the learned models that are not present in \NI{\Abasic}{T}{\mb{S}}.%
% \footnote{Technically, we use a \emph{weak trace preorder} but the precise treatment of this notion is outside the scope of the paper. We refer the interested reader to, e.g,~\cite{cleaveland2001equivalence}.}
\footnote{Technically, we implement this check using the \emph{weak trace preorder} of the mCLR2 model checker. See~\cref{sec:sancus} for more details.}
We call such differences in the observable behavior \emph{witnesses of interfering behavior}.

A reader may wonder why we need multiple silent observables rather than just one.
The idea is that $\mc{O}_\tau$ should define as silent those observables that are useful for a potential user to reconstruct an attack, but that may not be visible to the attacker, and thus whose appearance should not be considered interfering. % and should be ignored by the model checker.
Building on~\cref{thm:ideal-correctness}, it is trivial to prove that in the ideal case we verify \PNI{\Abasic}{\Aadvanced}{T}{\mb{S}} or provide a counterexample to it.

What can we say about comparing the models learned via $(\delta, \epsilon)$-\Lsharp for given $\delta$ and $\epsilon$?
\begin{theorem}\label{thm:pac-correctness}
Assume that for any given $s \in \mb{S}$ the behavior of $\SA{\Abasic}{T}{s}$ and of $\SA{\Aadvanced}{T}{s}$ are expressible as a Mealy machines, that $\mc{D}$ is the probability distribution on the set of possible action sequences according to $A$ and $T$, and that the sampling in the PAC oracle is performed according to $\mc{D}$.
Also, assume that for any $s \in \mb{S}$, $(\delta, \epsilon)$-\Lsharp learns \modelA{\Abasic}{T}{s} and \modelA{\Aadvanced}{T}{s} for \SA{\Abasic}{T}{s} and \SA{\Aadvanced}{T}{s}, respectively.
Let $\mc{C}_{(\Abasic, \Aadvanced, T)} \triangleq \forall s_0, s_1 \in \mb{S}.\ \modelA{\Abasic}{T}{s_0} \oequiv \modelA{\Abasic}{T}{s_1} \Rightarrow \modelA{\Aadvanced}{T}{s_0} \oequiv \modelA{\Aadvanced}{T}{s_1}$, i.e., a predicate that checks~\cref{def:pres-NIA} on learned models.

If $\mc{C}_{(\Abasic, \Aadvanced, T)}$, then with probability at least $(1-\delta)^4$ the probability that the checking phase does not make any error is at least $(1-\epsilon)^4$.
\end{theorem}
\newcommand{\Err}{\ensuremath{\mb{E}_{(\Abasic, \Aadvanced, T)}}}
\newcommand{\Errc}{\ensuremath{\mb{E}^c_{(\Abasic, \Aadvanced, T)}}}
\newcommand{\rc}{\ensuremath{c_{A, T, s}}}
\begin{proof}
    % Let $\Corr$ be the event ``\modelA{A}{T}{s} accepts the same language as $m(A, T, s)$'', the probability of the event $\Corr$ is expressed by a continuous random variable in the $[0, 1]$ interval.
    As a consequence of~\cref{thm:pac-pni} and since different PAC learning executions are one independent of the other, the probability that all four learned models involved in evaluating $\mc{C}_{(\Abasic, \Aadvanced, T)}$ are $\epsilon$-correct is $\geq (1 - \delta)^4$.

    Assume now that all these models are $\epsilon$-correct, which is the probability of making an error?
    We say that our checking procedure makes and error whenever $\mc{C}_{(\Abasic, \Aadvanced, T, {\mb{S}})}$ and $\lnot \PNI{\Abasic}{\Aadvanced}{T}{\mb{S}}$ or $\lnot\mc{C}_{(\Abasic, \Aadvanced, T, {\mb{S}})}$ and $\PNI{\Abasic}{\Aadvanced}{T}{\mb{S}}$.
    Let $\Err$ be the event occurring in case of such an error, and let $\Errc$ be its complementary event.
    % Computing $P(\Err)$ exactly would require knowledge of the probability of violations to $(\PNI{\Abasic}{\Aadvanced}{T})$, however we can, once again, give an upper bound.
    It is easy to see that $\Errc$ occurs at least when all the learned models are equal to their respective system.
    Thus, we can conclude that $P (\Errc) \geq (1-\epsilon)^4$ with probability at least $(1 - \delta)^4$.
\end{proof}
%
% In our example the attacker can only observe $\mi{TimerA}$ when untrusted code is executed, so we remove any other observable from the traces. --- confonde, tolto
% In particular, we remove all observables when the execution is inside the enclave and we simply write $\tau$ to represent a silent observable.
% For instance, it is typical for \mc{O} to mark as silent the $\mi{Time}$ observable, otherwise the attacker would be able to measure the duration of each single instruction executed in the enclave, giving an unrealistic threat model. The attacker, in practice, can measure time when the execution is transferred back to them through the $\mi{JmpOut}\ k$ action.\ricS{correct?}
% an explicit counter which is checked when the enclave terminates, as done in \cref{ex:attacker} or, in mode advanced threat models, using interrupts as we will discuss in \cref{sec:sancus}.
%
% In this way, only differences in behavior inside an enclave will be considered as interfering by the model checker.
%

%
%
\myskip
\begin{example}[Checking phase]
\label{ex:checking}
For this simple running example we illustrate our check on \NI{\Abasic}{T}{\mb{S}} but in~\cref{sec:sancus} we will leverage this functionality to check if \PNI{\Abasic}{\Aadvanced}{T}{\mb{S}}.
%(\cref{prop:full} of \cref{sec:intro}).
Here, the attacker can only observe $\mi{TimerA}$ when untrusted code is executed and every other observable is silent.
Notice that, all attack actions are the same, since the attacker is deterministic and shared by the two models.
% so we remove any other observable from the traces.
% In particular, we remove all observables when the execution is inside the enclave and we simply write $\tau$ to represent a silent observable.
Performing the checking phase over \modelA{\Abasic}{T}{0} and \modelA{\Abasic}{T}{1} of~\cref{ex:learning}, we obtain the counterexample to noninterference of~\cref{fig:tree-toy},
% \flaS{tagliata la caption Fig 3}
represented as a \emph{witness graph} in which arrows in common for both models are black, solid, those of \modelA{\Abasic}{T}{0} are red, dashed and those of \modelA{\Abasic}{T}{1} are blue, dotted.
In the witness graph, the violation to \NI{\Abasic}{T}{\mb{S}} is clearly visible: a different time taken by the enclaved execution measured by \mi{TimerA}: 12 on the left and 5 on the right.
% From this example it should be clear why we use both $\mc{O}$ and $\mc{O}_\tau$: the former, richer set of observables, is used to spot relevant differences in the execution such as the \lstinline|cmp| of the unbalanced branch which depends on the secret input.
% These differences are not necessarily attacks, as they are not observable by the attacker, but are useful to understand why the attacker is distinguishing the two models.
% ßThe only observable actions that distinguish the two model, from the attacker point of view, are in fact the $\mi{TimerA}\ k$ measuring the overall execution time of the enclaved component.
Since time is different for the two secret inputs $0$ and $1$ the attacker can actually infer such a secret value.

% Notice that,~\cref{fig:tree-toy} reports a simplified version of the actual witness graph, since our algorithm  correctly detects that one could repeat the \Ar{} attack twice, as \lstinline|rst_nz| resets the CPU.
% This is useful for less constrained attackers that could exploit previous observations to adapt their attack technique after rebooting.
% We refer the interested reader to the additional material~\cite{additionalmaterial} for the complete witness graph.
% Note that, swapping $\mu$ and $\mu'$ produces an empty witness graph: there is no behavior in $\mu_{\mathghost}$ that cannot be reproduced in $\mu_{\skull}$.
% The tool successfully detected the differences in the behavior of the enclave when under attack, as it was predictable by looking at the models in~\cref{fig:learned-toy}.
\end{example}
%
%In the next section, we give an in-depth description of our tool which implements and extends the above technique to automatically detect interfering behavior in Sancus. \flaS{tagliare questa frase?}
% Our tool enables us to systematize and reproduce most of the results by~\citet{mindthegap} and . \flaS{Nuovi attacchi?}

%-------------------------------------------------------------------------------
\section{Bridging the Gap in Sancus}\label{sec:sancus}
%-------------------------------------------------------------------------------
% !TEX root = ../main.tex
In this section we report the results of our analysis of the Sancus architecture using the method of~\cref{sec:overview}.
Our goal is to systematically analyze the architecture in a black-box fashion and to discover possible issues or to prove it secure, both in past and in current versions of the architecture.
% To this end, we automatically learn and compare different formal models of Sancus behavior under attack, looking for any violations to~\cref{def:NI}.
To this end, we automatically learn and compare different formal models of Sancus behavior under attack, looking for any violations to~\cref{def:pres-NIA}.
% Instead of manually defining our formal models,
% We use active automata learning to learn many of them, considering a set of attackers interacting with different systems.
All the attackers $A$ and enclaves $T$ will be defined in the format of~\cref{sec:specification}.
The system under analysis will consist of a specific Sancus version running $\SA{A}{T}{s}$ for some $s \in \mb{S}$.
The \toolname tool and all specifications and results are available at~\cite{additionalmaterial}.

\subsection{The \toolname Tool}\label{sec:alvie}
\toolname (\toolacron) is a tool written in \ocaml that implements the method of \cref{sec:overview}, allowing users to specify attackers, learn models of behavior of the system under attack, compare the learned models, and synthesize executable Sancus programs directly from the attack traces found during comparisons.
Our experiments use the learning algorithm $(\epsilon, \delta)$-\Lsharp with $\epsilon = \delta = 0.01$ and were performed on a server equipped with an Intel(R) Xeon(R) E5-2699 v4 @ 2.20GHz with 42 CPUs and 128GB of RAM.
On our machine, running the experiments of this section takes about an hour for simpler specifications and a couple of days of bigger/complex instances.
This relatively short learning time comes from the fact that \toolname spawns multiple processes, one for each learning and checking task, making full use of the multicore features of our machine.

The development of \toolname posed some nontrivial challenges that we summarize below.

\myskip
\subsubsection*{Extracting system behavior}
We need to implement $\mc{B}$, observing and extracting a sequence of observables $\mc{O}$ from concrete executions of the Sancus simulator.
For this, we first take the input sequence (i.e., the sequence of attacker/system interactions that we want to study), and synthesize a Sancus program.
Then, we compile the program into binary code, which is loaded and executed by a Sancus simulator obtained using the Verilator tool~\cite{verilator}.
Verilator generates a cycle-accurate simulator from the real Verilog code of Sancus so, assuming Verilator's correctness, our results should be identical to the ones obtained interacting with an FPGA implementation, except for the performance.
The execution of the simulator produces a \emph{Value Change Dump} (VCD) file that our tool
% first parses
%using the Python library \texttt{Verilog\_VCD}~\cite{verilogvcd} (wrapped inside an \ocaml module)
% , and then
analyzes in order to extract the abstract behavior of the input sequence.
% Correctly abstracting the signals appearing in the VCD file into our abstract behavior was not trivial.
% In particular, a certain amount of reverse-engineering was required to track CPU mode changes, from the attacker code to  the enclave code, and vice versa.
We also need to keep track of abstract input actions compiled into multiple assembly instructions that may be interrupted by an attack.
% Indeed, when one of these actions interrupted by an attacker, their execution may not be complete and -- when it is -- is dispersed at different points in time (sometimes even spread across multiple CPU mode switches).
To this aim, we tag each instruction composing an input action with a unique label and collect them into a \emph{pending} list.
When that specific instruction is executed, we add its abstract observable (derived from the contents of the VCD file) to the observed behavior and retire it from the list of pending labels.

\myskip
\subsubsection*{Optimization of the learning phase}
Efficiently learning the formal model is another demanding task.
% The first fundamental question was to choose the right algorithm for active automata learning algorithm.
To the best of our knowledge, the \ocaml ecosystem lacks an active automata learning library.
Thus, we decided to implement our own and chose the state-of-the-art algorithm \Lsharp~\cite{vaandrager2022new}, since it is relatively easy to understand and implement, especially with the help of the \ocaml type system.
% Furthermore, \Lsharp was easy to adapt to our needs, and we managed to made it efficient enough for our use case.
Besides implementing the algorithm with efficient data structures, we have optimized its exploration of the input space by binding it to $A$ and $T$.
In particular, as explained in~\cref{sec:learning}, at each step the algorithm can choose to provide an input to the system and observes its output: the choice is guided by the knowledge of which inputs are legal and which are not, depending on the latest observations on the system and on the specifications of the attacker and the enclave.
Moreover, during the learning phase, we skip the so-called non-determinism check: if a given input trace has already  been executed on the simulator, its observable behavior is cached and reused unchanged instead of repeating the execution to validate that the CPU behaves deterministically.
In some cases this optimization has resulted in a $60\times$ speedup.
E.g., learning the models for the example in~\cref{sec:overview} now takes about $15$ seconds on our machine, before optimization it took about $900$ seconds.
Skipping the non-determinism check requires careful handling of the tool's internal data structures that are used during the learning to keep track of the internal state of the CPU.

\myskip
\subsubsection*{Generation of witness graphs}
We designed and implemented witness graphs that are easy to visualize and to interpret.
As illustrated in \cref{ex:checking}, witness graphs are built starting from the witnesses of violations to~\cref{def:pres-NIA} given by the mCRL2 model checker.
To collect all these witnesses one could think of invoking the model checker several times and integrating the results.
% This way, these two models would not present the difference that gave rise to the violation and the model checker may look elsewhere.
However, the number of calls to the model checker required to successfully collect all these witnesses may be too large to be manageable (or even infinite, for example, in the presence of resets).
To reduce this complexity, when a witness $W$ is found, we remove all the witnesses longer than $W$, and keep $W$ as a representative.
Intuitively, $W$ corresponds to the shortest attack of ``its kind''.
%Interestingly, w
We also provide a tool that synthesizes executable code from witnesses of interfering behavior, thus allowing users to collect the witnesses and use them as regression tests when the version of the system changes.
% (cf. \cref{sec:appendix_synthesis}).
We use the following notation: if an action is performed by the attacker it is preceded by \texttt{att}, otherwise it is preceded by \texttt{enc}.
For instance, in the case of the witness graph of~\cref{fig:tree-toy}, the witnesses of interfering behavior is automatically synthesized as follows:
\begin{lstlisting}
att: start_counting 256;
att: create <enc_s, enc_e, data_s, data_e>;
att: jin enc_s;
enc: cmp s, #0;
enc: ubr
\end{lstlisting}
\toolname is then able to automatically substitute \lstinline|s| with the secret, compile the code to assembly, run it on the simulator, and collect the results.
For example, it can correctly detect the difference between the two enclaves when it executes the above code: if \lstinline|s|$\ =0$ the output is $\mi{JmpOut}\ 2,\, \mi{TimerA}\ 12$, whereas if \lstinline|s|$\ =1$ leads to $\mi{JmpOut}\ 2,\, \mi{TimerA}\ 5$.

For the synthesis of the other attacks and anomalies, we refer the interested reader to the files \texttt{attack.ml} and \texttt{attacklib.ml} in the \texttt{alvie/code/test} folder of~\cite{additionalmaterial}.

% \subsection{Using \toolname on Sancus}
\subsection{Attackers and Enclaves in Sancus}\label{sec:attackers_enclaves}
In this section we specify various attackers and enclaves for Sancus that we will use in the next section to automatically discover attacks.

\myskip
\subsubsection*{Attackers capabilities}
We consider a list of attacker capabilities that
% We follow a top-down approach and begin by defining the basic capabilities of the attacker we consider for our study.
% These capabilities
are quite standard in the context of Trusted Execution Environments~\cite{vrased,sancus,sgx}, and merge ideas from formal methods and more applied studies~\cite{mindthegap,busi2021securing,sancus}.
%\ricS{spostata qui la citazione da sotto}
% The capabilities of our attackers are as follows,
% In principle, each section of attacker's specification can choose to use a subset of them.
%
\begin{itemize}
    \item\textbf{C1: Creation of enclaves.}
        Creating enclaves is of course fundamental for interacting with them.
        The action for creating enclaves is \createenclave{}, where \encs{} and \ence{} respectively denote the start and the end of the code of the enclave, and \datas{} and \datae{} respectively denote the start and the end of its data.
        All the locations in the ranges [\encs, \ence) and [\datas{}, \datae{}) are protected from external access (read/write/execution), except for \encs{} which is the entry point of the enclave and attackers can jump to it (cf. capability \textbf{C2}).
        % The location $cs$ is the entry point of the enclave, and attackers can jump to it (cf. capability \textbf{C2}).
        % All the other locations in the range $[cs, ce)$ are protected from external read and writes, and cannot be jumped to directly.
        % Furthermore, the access control mechanism of Sancus completely forbids external access (no read/write/execution) to the data in the range $[ds, de)$. \flaS{de escluso?}
        To improve readability, in the examples we always use \lstinline|create_enclave| to denote \linebreak \createenclave.
    \item\textbf{C2: Starting enclaves with parameters.}
        Another crucial capability is to start an enclave and possibly pass data to it.
        Data is exchanged in Sancus via registers, so an attacker can use any instruction operating on registers, e.g., \lstinline|mov| or \lstinline|add|.
        Starting an enclave (called the \emph{jump in} operation) is performed by jumping to the beginning of its code section.
        % , i.e., if the enclave was created with \createenclave the jump in can be performed using \lstinline|jmp #$\encs$|.
        If the enclave was created with \createenclave{} the jump in can be performed using the \lstinline|jin $\encs$| abstract action.
    \item\textbf{C3: Access to timer, registers and memory.}
        MSP430 architecture (and thus Sancus) comes with a cycle-accurate and user-programmable timer called \emph{Timer A}.
        An important attacker's capability is the access to Timer A.
        This capability is implicitly offered to the attacker through the $\oequiv$ relation as shown in~\cref{sec:checking}.
        We also let the attacker monitor registers and unprotected memory locations so to detect direct information leakages from the enclave.
        % \matteo{TODO: this is where we should explain abstract behavior!}
        For that, we extend the set of observables $\mc{O}$ we introduced in~\cref{ex:abstract} as follows:
        \begin{itemize}
            \item $\mi{Diverge}$, denotes that the CPU loops indefinitely.
            Since this is not decidable in general, we assume the CPU is looping  after a predefined timeout;
            \item $\mi{Handle}\ k$, denotes that the CPU detected an interrupt which took $k$ cycles to start its handler;
            \item $\mi{Reti}$, denotes  that the attacker has finished executing the interrupt handler and is giving back the control to the enclave;
            \item $\mi{Reset}$, denotes a CPU reset;
            \item $\mi{GIE}$, denotes that interrupts are enabled after executing the input;
            \item $\mi{UMem}\ v$, denotes
             the value of a predefined memory location  controlled by the attacker after executing the input;
            \item $\mi{Reg}\ v$, denotes the value of a predefined register after executing the input;
            \item $\mi{UM}$ or $\mi{PM}$, denotes the CPU mode after executing the input: $\mi{UM}$ means that the CPU is in Unprotected Mode and the attacker controls it, $\mi{PM}$ stands for Protected Mode and it means that the enclave controls the CPU.
        \end{itemize}
        % Attackers willing to use this capability can simply access to the \lstinline|TAR| address using classical assembly instructions.
        % For instance, \lstinline|mov &TAR, r5| copies the current value of the timer counter into the \lstinline|r5| register.
    \item\textbf{C4: Interrupts from a cycle-accurate timer.}
        Attackers with this capability may use the action \lstinline|timer_enable k| to start a timer that raises an interrupt \lstinline|k| CPU cycles after the enclave has been started.
        These attackers can also specify how to handle scheduled interrupts using the \lstinline|isr| section of their specification.
    \item\textbf{C5: Returning from an interrupt.}
        Returning from interrupts is performed through the \lstinline|reti| instruction.
    \item\textbf{C6: Abusing enclave entry points.}
        The attacker might try to enter the enclave in unexpected ways by performing \textbf{C2} and \textbf{C5} in the wrong sections.
        Typical examples are: starting the enclave while handling an interrupt or returning from an interrupt out of the \lstinline|isr| section.
    % \item[\textbf{C3: Resetting/halting the CPU.}]
    %     Resetting and halting the CPU can be done in two ways: \lstinline|rst| that resets the system, and \lstinline|rst_nz| that halts the system if the previous compare instruction \lstinline|cmp| returns \lstinline|0|, i.e., if the compared values are equal, and resets otherwise.\ricS{Credo che useremo solo rst controllare}

        % As we explain below, this capability may seem unexciting, but it can be used in surprisingly interesting ways.

    % \item[\textbf{C7: Adapting to enclave behavior.}]\ricS{Questa non so se ci serve più}
    %     The last capability we consider for our attackers is their ability to modify their behavior based on  their observations of the enclaves.
    %     Typically, this includes measuring the execution time of an enclave, reading the return values of an enclave (stored in registers), monitoring the attacker's own data memory, or observing whether the interrupt handler was executed or not.
    %     Indeed, by comparing observed events and values with the expected ones, an attacker may be able to discover novel issues in enclaves.
    %     For example, like in the~\cref{sec:overview} example, an attacker may discover that an enclave's execution is not constant-time by measuring its execution time.
    %     Also, an attacker may try to interfere with an enclave in various ways (e.g., using \textbf{C2} or \textbf{C5} in unintended ways), and look for unexpected values returning from the enclave,  thereby revealing possible integrity issues.
\end{itemize}
%
% We selected these capabilities inspired by previous work on the Sancus architecture~\cite{mindthegap,busi2021securing,sancus}, however this set of capabilities is relatively easy to extend, if necessary, by specifying new abstract actions for the attacker and how they are compiled to assembly. \flaS{con questa frase non e' che ci dicono che ne abbiamo specificate poche o  non tutte?}
 % when necessary (roughly done by specifying its syntax and how it can be compiled down to assembly).

We now
%proceed to
define a very basic attacker and show how to
%progressively
add more capabilities to represent  interesting attack scenarios.
%
% \begin{figure}
% \begin{lstlisting}
% isr {
%     eps                 /* empty section */
% };

% prepare {
%     start_counting 256; /* start a 8-bit timer */
%     create_enclave;     /* create the enclave  */
%     jin $\encs$           /* start the enclave   */
% };

% cleanup {
%     eps                 /* empty section */
% };
% \end{lstlisting}
% \caption{The basic attacker $\Abasic$ which creates and starts an enclave, observing its behavior in terms of time, registers and unprotected memory locations. Observations are part of $\oequiv$ definition. }
% \label{fig:basicattacker}
% \end{figure}

\myskip
\subsubsection*{Basic attacker}
The basic attacker $\Abasic$ of~\cref{fig:basicattacker}, discussed in~\cref{ex:attacker},  can create enclaves and start them, observing the execution time, registers and unprotected memory locations.
It
%is identical to the one of~\cref{ex:attacker} and
has capabilities \textbf{C1}, \textbf{C2} and \textbf{C3}.
% , has an empty \lstinline|isr| section and can reset/halt the CPU based on the behavior of the enclave in the \lstinline|cleanup| section (\textbf{C3} and \textbf{C7}).
This attacker can  detect timing side-channels (cf.~\cref{ex:checking}), direct leaks of the secret such as copying of the secret to a public register or to unprotected memory locations, and other indirect leaks such as executing a command with an observable side effect in one of two branches.
This attacker does not use interrupts so, once the enclave is started there is no way for the attacker to affect its execution in any way.
We use this attacker as a baseline for secure enclaves: if an enclave is noninterfering with respect to this attacker, then we regard it as secure and we expect it to be noninterfering even with respect to more powerful attackers, along the full abstraction results of~\citet{busi2021securing}.
\begin{figure}[tb]
    \begin{lstlisting}
enclave {
    mov s, &unprot_mem; /* direct leakage of the secret */

    |

    /* indirect leakage if c has observable side effect */
    cmp s, #0;
    (
        ifz (c) (nop$^{\textrm{clocks(c)}}$)
    );
};
    \end{lstlisting}
    \caption{\small
        %An i
        Insecure enclave.
        %exhibiting leakages that are under the responsibility of the programmers.
        %The leakages are all  detected by  the basic
        % attacker $\Abasic$ of \cref{fig:basicattacker}
        % detects leakages.
    }
    \label{fig:insecureenclave}
\end{figure}
\subsubsection*{An insecure enclave}
\cref{fig:insecureenclave} illustrates a generic enclave specification that exhibits insecure behaviors with respect to $\Abasic$.
This is useful to clarify the scope of the analysis.
In fact, these insecure behaviors should not be considered attacks as they are not expected to be fixed by the Sancus security architecture.
In particular, the \lstinline|mov| is a direct leakage copying the secret to unprotected memory.
This is allowed in Sancus, so it is up to the programmer to avoid it in the code.
Command \lstinline|ifz (c1) (c2)| specifies an if-then-else branch based on the previous comparison on the secret \lstinline|s|.
The enclave of \cref{fig:insecureenclave} uses the construction \lstinline|ifz (c) (nop$^{\textrm{clocks(c)}}$)| that has a command \lstinline|c| on one side and a sequence of length \lstinline|clocks(c)| of \lstinline|nop| commands on the other side, where \lstinline|clocks(c)| is the expected number of clock cycles of \lstinline|c|.
This can be easily obtained from the processor's documentation~\cite{ti-msp430,openMSP430}.
Since \lstinline|nop| does nothing and takes exactly one clock cycle, this construction prevents the timing side-channel of~\cref{ex:timing1}, by letting the two branches take the same amount of clock cycles.
However, the branch is insecure for all commands \lstinline|c| that have observable side effects, i.e., that are not equivalent to a long enough sequence of \lstinline|nop| commands.
For example, \lstinline|mov #42, r4| writes a constant to register \lstinline|r4| and is insecure, since values of registers are returned to the attacker at the end of the enclaved execution.
Similarly, \lstinline|mov #42, &unprot_mem| is insecure as it writes a constant to an unprotected memory location.
In these cases, the attacker can infer the value of the secret by observing whether $42$ has been written or not.
These attacks are detected by $\Abasic$ through capability \textbf{C3} which enables the attacker to observe register values and unprotected memory locations through a suitable $\oequiv$ definition.
\begin{figure}[tb]
        % /* secure if c has no observable side-effect */
        % ifz (c) (nop$^{\textrm{clocks(c)}}$)
        % |
\begin{lstlisting}
enclave {
    cmp s, #0; /* compare s with 0 */
    (
        /* secure even if c has observable side-effect */
        ifz (c; nop) (nop; c)
    );
};
\end{lstlisting}
\caption{\small A secure enclave exhibiting two kinds of balanced branches that are intuitively secure.}
\label{fig:secureenclave}
\end{figure}

\begin{figure}[tb]
% \centering
% \begin{lstlisting}
% enclave {
%     cmp s, #0;
%     (
%         /* ifz (c) (nop$^{\tt \color{ForestGreen}{clocks(c)}}$) */
%         ifz (mov r5, r5) (nop) |
%         ifz (mov &$\encs$, &$\encs$)
%             (nop; nop; nop; nop; nop; nop; nop)

%         |

%         /* ifz (c; nop) (nop; c) */
%         ifz (mov &unprot_mem, r8; nop)
%             (nop; mov &unprot_mem, r8) |
%         ifz (mov #42, &unprot_mem; nop)
%             (nop; mov #42, &unprot_mem) |
%         ifz (jmp #$\datas$; nop)
%             (nop; jmp #$\datas$) |

%         ifz (dint; nop) (nop; dint) |
%         ifz (rst; nop) (nop; rst)
%     );
% };
% \end{lstlisting}
\begin{lstlisting}
enclave {
    cmp s, #0;
    (
        /* ifz (c; nop) (nop; c) */
        ifz (mov r5, r5; nop) (nop; mov r5, r5) |
        ifz (mov &enc_s, &enc_s; nop)
            (nop; mov &enc_s, &enc_s) |

        ifz (add #1,  &data_s; nop) (nop; add #1,  &data_s) |
        ifz (mov #42, &data_s; nop) (nop; mov #42, &data_s) |

        ifz (mov &unprot_mem, r8; nop)
            (nop; mov &unprot_mem, r8) |
        ifz (mov #42, &unprot_mem; nop)
            (nop; mov #42, &unprot_mem) |
        ifz (jmp #data_s; nop)
            (nop; jmp #data_s) |

        ifz (dint; nop) (nop; dint) |
        ifz (rst;  nop) (nop; rst)
    );
    (eps | mov &data_s, r4);
    jmp #enc_e
};
\end{lstlisting}
\caption{\small Instantiation of the secure enclave of~\cref{fig:secureenclave} used in our tests.
%The instantiation of the secure enclave of~\cref{fig:secureenclave} that we used in our tests.
}
\label{fig:generalenclave}
%         ifz (add #1, &$\datas$; mov &$\datas$, r4) (nop; nop; nop; add #1, &$\datas$; mov #1, r4) |
\end{figure}

\subsubsection*{Secure enclave}
We now define a generic enclave that is intuitively secure with respect to the basic attacker $\Abasic$.
The idea is to devise representative code templates that we expect to be protected by the Sancus enclave.
We use a branch over the secret and combine various commands so that the branches will always be balanced in time since, as we discussed in~\cref{sec:overview}, unbalanced branches are clearly insecure in Sancus.
\cref{fig:secureenclave} reports our general scheme.
% We basically use two kinds of branches.
% The first one is the same as the one used in the insecure enclave of~\cref{fig:insecureenclave}, but here command \lstinline|c| is assumed to have no observable side effects.
% If this is the case, \lstinline|c| should be indistinguishable from a sequence of \lstinline|nop| commands taking the same amount of clock cycles.
The construction has the form \lstinline|ifz (c; nop) (nop; c)|.
 % \lstinline|c; nop| on one side and \lstinline|nop; c| on the other side.
%Here w
We check that changing the ordering of commands does not leak the current branch and thus the secret even when \lstinline|c| has observable side effects.
Interestingly, for uninterruptible enclaves, even if \lstinline|c| has a visible side effect, such as writing a constant value into a public memory location, the two branches are not distinguishable, since the attacker will only be able to see the change after the jump out from the enclave.
In \cite{busi2021securing} it is proved that the same should hold with interruptible enclaves, but we will see that this is not the case for the latest Sancus commit.

This scheme cannot be exhaustive in covering all the interesting secure behaviors that one could check, but we will see that it is expressive enough to capture most of the documented attacks, plus two new ones.
Notice also that, the scheme can be instantiated with many different commands.
We used the instantiation reported in~\cref{fig:generalenclave} for our tests as it offers good generality by covering different interesting cases with a reasonable learning time, as we will discuss next.
Intuitively: the first two commands move a value over itself which clearly should have no side effects (lines 5 and 6); the next two commands perform operations on \datas{} which is not accessible to the attacker (lines 9 and 10), but notice that \datas{} is non-deterministically copied to \lstinline|r4| which is returned to the attacker (line 22); the next three commands try read/write/execute violations (lines 12, 14 and 16); finally, the next  commands try to manipulate the interrupt register and reset the CPU (lines 19 and 20).
% Since we compare \secret with \lstinline+0+, it is enough two secret values and we let $\mb{S} = \{ 0, 1 \}$.
Since we compare \secret with \lstinline+0+, two secret values are enough to cover all the cases and we let $\mb{S} = \{ 0, 1 \}$.

% All the branches check that the timing of these operations is not observable by prefixing or appending a \lstinline|nop|.
% The first three branches are examples of instructions with duration $1$, $1$ and $7$ (resp.) that should not have attacker-visible side effects.
% The next three branches perform accesses that should be forbidden by Sancus: reading from unprotected memory, writing to unprotected memory, and jumping into the protected data section.
% The last two branches are similar but perform the check for particular actions: \lstinline|dint| tries to manipulate the flag that enables and disables interrupts, which should be forbidden from the enclave, and \lstinline|rst| resets the CPU.
% Of course it is possible to extend the enclave much more but the one presented here offers a reasonable generality and learning time, as we will discuss next.
%
\begin{figure}[tb]
\begin{lstlisting}
isr {
    /* set interrupt at different times */
    (timer_enable 0 | ... | timer_enable n);
    (reti | jin $\encs$) /* return or start the enclave */
};

prepare {
    /* set interrupt at different times */
    (timer_enable 0 | ... | timer_enable n);
    create_enclave;  /* create the enclave   */
    jin enc_s        /* start the enclave    */
};

cleanup {
    (eps | reti)    /* nothing or return from interrupt */
};
\end{lstlisting}
\caption{\small The advanced attacker $\Aadvanced$ with all capabilities. }
\label{fig:advancedattacker}
\end{figure}

\subsubsection*{Advanced attacker}
In~\cref{fig:advancedattacker} we present an advanced attacker $\Aadvanced$ with all the capabilities.
It extends $\Abasic$ by: setting an interrupt at various times, from $0$ to $n$, using command \lstinline|timer_enable k|, both in \lstinline|prepare| and in \lstinline|isr| sections (\textbf{C4}); returning from an interrupt with \lstinline|reti| in \lstinline|isr| (\textbf{C5}) but also in inappropriate sections, such as \lstinline|cleanup| (\textbf{C6}); starting the enclave while handling an interrupt in the \lstinline|isr| section (\textbf{C6}).
We use this attacker and the secure enclave of~\cref{fig:generalenclave} to automatically rediscover known attacks and to find new ones.

\newcommand{\vbonedescr}{The first instruction after a \texttt{reti} takes one additional cycle}
\newcommand{\vbtwodescr}{There exists an instruction with execution time $>6$}
\newcommand{\vbthreedescr}{Enclaves can be resumed multiple times with \lstinline|reti|}
\newcommand{\vbfourdescr}{Restarting an interrupted enclave from the interrupt-service routine is allowed}
\newcommand{\vbfivedescr}{Creation of multiple enclaves is allowed}
\newcommand{\vbsixdescr}{Enclaves can access unprotected memory}
\newcommand{\vbsevendescr}{Enclaves can manipulate interrupt behavior}
\newcommand{\vconedescr}{The formal model misses DMA peripherals}
\newcommand{\vctwodescr}{The formal model misses the watchdog timer}

\newcommand{\vbeightdescr}{Read/Write violations reset the CPU}
\newcommand{\vbninedescr}{The enclave can reset the CPU explicitly}

\begin{table}[tb]
    \resizebox{\columnwidth}{!}{
    \begin{tabular}{@{}ll@{}}
    \toprule
        % & Name & Discrepancy\\ \midrule
    % \multirow{10}{*}{\centering \textbf{IMs}} &
         \textbf{V-B1} & \vbonedescr \\
         \textbf{V-B2} & \vbtwodescr \\
         \textbf{V-B3} & \vbthreedescr \\
         \textbf{V-B4} & \vbfourdescr\\
         \textbf{V-B5} & \vbfivedescr\\
         \textbf{V-B6} & \vbsixdescr\\
         \textbf{V-B7} & \vbsevendescr\\        \midrule

         \color{Red} \textbf{V-B8}  & \color{Red} \vbeightdescr\\
         \color{Red} \textbf{V-B9}  & \color{Red} \vbninedescr\\
        % & \color{Red} \textbf{V-B9} & \color{Red} \vbtendescr\\

        % \multirow{2}{*}{\centering \textbf{MAs}} &
        %   \textbf{V-C1} & \vconedescr\\
        % & \textbf{V-C2} & \vctwodescr\\
        \bottomrule
    \end{tabular}
    }

    \caption{\small The new discrepancies found in this paper (in red) together with the  implementation-model mismatches of \cite{mindthegap} (in black).}\label{tab:mindthegap}
\end{table}

\subsection{Rediscovering Known Attacks}
\label{sec:knownattacks}
We use the attackers and enclaves presented in the previous section to systematically reproduce the \emph{implementation-model mismatches} presented by~\citet{mindthegap}.\footnote{
The attacks of \cite{mindthegap} based on \emph{missing attacker capabilities} are out of the scope of our analysis.}
They are
% In this work we focus on
% \citet{mindthegap} distinguish two classes of discoveries: \emph{implementation-model mismatches} (IMs), and \emph{missing attacker capabilities} (MAs). \flaS{Se togliessimo V-C1 e V-C2?}
% We focus on the forme
 discrepancies between the Sancus implementation and its formal model, \sancusv~\cite{busi2021securing}, leading to attacks that can be performed on the implementation, but are unsuccessful in the model.
We summarize them in the top of~\cref{tab:mindthegap}, in black color.
\textbf{V-B1} shows that the first instruction following a \lstinline|reti| takes one additional cycle in the implementation while
\textbf{V-B2} identifies an instruction (a \lstinline|mov| between locations in the code memory) that takes $7$ cycles, contrary to the assumption in \sancusv that the maximum duration of instructions is $6$ cycles. These two mismatches break the padding mechanism of secure interruptible enclaves in \sancusv.
\textbf{V-B3} shows that enclaves can be resumed multiple times using a \lstinline|reti|, and \textbf{V-B4} shows that restarting an interrupted enclave from the interrupt-service routine is permitted in the implementation.
Both actions are forbidden in \sancusv as they break the intended control-flow on enclaved executions.
\textbf{V-B5} concerns the possibility of creating multiple enclaves, an action not supported by the model.
Finally, \textbf{V-B6} and \textbf{V-B7} show that the access-control policy given by \sancusv is not correctly implemented.
% MAs denote instead discrepancies between the actual capabilities of attackers and those modeled in the formal semantics. \textbf{V-C1} shows how to exploit unmodeled DMA peripherals to carry out an attack in the Sancus implementation, and \textbf{V-C2} does the same with the unmodeled watchdog timer.
These attacks were found on the original Sancus commit and fixed in next commits.\footnote{The original Sancus commit is \origcommit while the latest commit including all the mitigations is \lastcommit, both available at \cite{mindthegap-repo}.}

% For further details, we refer the interested reader to the original paper.
%
% The proof in \sancusv is based on full abstraction, showing that secure interruptible enclaves do not affect the attacker's ability to discriminate trusted components.
Consider the attackers \Abasic and \Aadvanced of~\cref{fig:basicattacker,fig:advancedattacker}, and the secure enclave T of~\cref{fig:generalenclave}.
% Full abstraction implies that
% %
% \begin{equation}%\label{prop:fullinstance}
% % \forall s_0,s_1 \in \secrets,
% \SAzero \oequiv \SAuno \mbox{ iff } \SApzero \oequiv \SApuno
% \end{equation}
% %
% This property is very appealing as it allows a developer to program \T{} without worrying about interrupts controlled by the attackers.
We use \toolname with the goal of proving that~\cref{def:pres-NIA} is broken in the original Sancus commit and it holds in the latest version.
As expected, the first result is that $\SA{\Abasic}{T}{0} \oequiv \SA{\Abasic}{T}{1}$ on both commits, i.e., the secure enclave of~\cref{fig:generalenclave} is indeed secure with respect to the basic attacker \Abasic.
% $\SAl \oequiv \SAr$ and $\SApl \oequiv \SApr$.
However, we obtain that $\SA{\Aadvanced}{T}{0} \not\oequiv \SA{\Aadvanced}{T}{1}$ on both Sancus commits, breaking preservation of noninterference.
While $\SA{\Aadvanced}{T}{0} \not\oequiv \SA{\Aadvanced}{T}{1}$ was expected on the original Sancus commit, due to the attacks of \cite{mindthegap} reported in \cref{tab:mindthegap}, the fact that equivalence does not hold on the latest Sancus commit indicates new undocumented attacks.
% The attacks detected by \toolname are examples of trusted components \T{} that are secure with respect to $A$, i.e., $\SAl \oequiv \SAr$, but are insecure with respect to $\Ap{}$, i.e., $\SApl \not\oequiv \SApr$, breaking full abstraction.
% Since the attacks correspond to real execution sequences they constitute  discrepancies between the formal model of \cite{busi2021securing} and its implementation.

\toolname found all the attacks of \cref{tab:mindthegap} except \textbf{V-B5}, since we have not yet implemented the ability to create multiple enclaves.
This is just an implementation issue that we plan to fix as future work.
% Moreover, we automatically analyzed the various Sancus commits, reconstructing the history of the fixes (cf. \cref{sec:appendix_history}).
Moreover, we automatically analyzed the various Sancus commits, reconstructing the history of the fixes and proving with given confidence and precision that some of those fixes are correct. %(see \cref{sec:appendix_history}).
%%%%%%%%%%%%%%%%%%%%%%%%%%%%%%%%
\cref{tab:summary} summarizes this process, highlighting the commits where the various patches occurred.
All the attacks (except \textbf{V-B5}) were found in the original commit and also in all the subsequent ones, up to the commit with the specific patch for that attack.
From then on, we verified that the attack was successfully prevented up to the last commit, i.e., no regression occurred in subsequent commits.
\begin{table}[t]
    \resizebox{\columnwidth}{!}{
        \begin{tabular}{@{}lccc@{}}
        \toprule
            % \textbf{}& \multicolumn{3}{l}{\textbf{Found violations?}}\\ \midrule
            \textbf{} & \textbf{Original commit (\origcommit)} & \textbf{Patch commit} & \textbf{Last commit (\lastcommit)}\\ \midrule
            \textbf{V-B1} & \xmark & \cmark\, (\vbonecommit) & \cmark\\
            \textbf{V-B2} & \xmark & \cmark\, (\vbtwocommit) & \cmark\\
            \textbf{V-B3} & \xmark & \cmark\, (\vbthreecommit) & \cmark\\
            \textbf{V-B4} & \xmark & \cmark\, (\vbfourcommit) & \cmark\\
            \textbf{V-B5} & --- & ---\, (\vbfivecommit) & ---\\
            \textbf{V-B6} & \xmark & \cmark\, (\vbsixcommit) & \cmark\\
            \textbf{V-B7} & \xmark & \cmark\, (\vbsevencommit) & \cmark\\ \bottomrule
        \end{tabular}
    }

    \caption{\small Validation of the fixes by \citet{mindthegap}: \cmark{} and \xmark{}  respectively indicate whether the attack is or is not prevented.}\label{tab:summary}
\end{table}

We explain how \toolname found \textbf{V-B6}.
The other attacks in \cref{tab:mindthegap} are captured in similar, but slightly more involved, ways (see \textbf{V-B1} in \cref{sec:appendix_witness} and the other ones in \cite{additionalmaterial}).
%\ricS{mostrare un altro witness o metterli in appendice?}
Consider the following command of the enclave used in our tests (cf. \cref{fig:generalenclave}):
\begin{lstlisting}[numbers=none,frame=none,xleftmargin=10pt]
ifz (mov #42, &unprot_mem; nop) (nop; mov #42, &unprot_mem)
\end{lstlisting}
Suppose that the interrupt is set after 1 clock cycle, i.e., in the middle of the \lstinline|mov| of the left branch and immediately after the \lstinline|nop| of the right branch.
The interrupt is postponed to the end of the current command and, moreover, the mechanism of secure interruptible enclaves adds an extra padding so that the execution of the \lstinline|isr| always happens after a fixed amount of clocks.
Thanks to this padding, the attacker cannot infer whether the execution is in the left or the right branch due to timing issues, however if the \lstinline|mov| of the left branch is successful the attacker will see value \lstinline|42| in \lstinline|unprot_mem| making it possible to distinguish the left branch from the right one.
The witnesses testifying this attack are reported in the witness graph in~\cref{fig:witness-b6}.
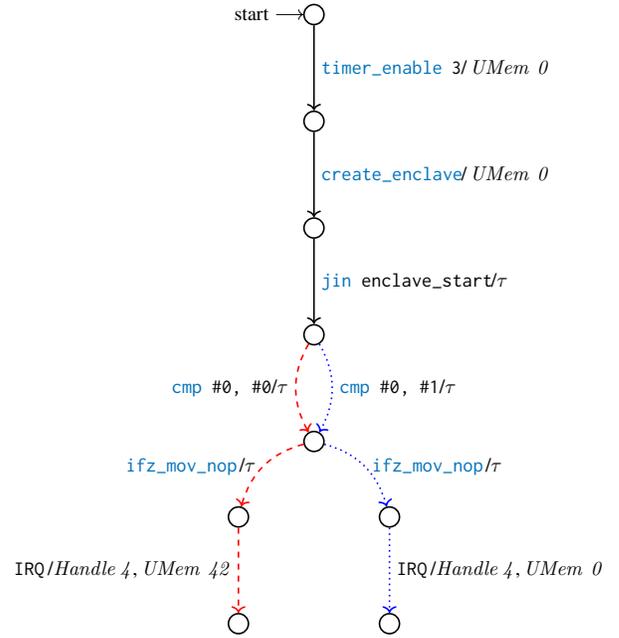
\begin{figure}[htb]
    \centering
    \resizebox{!}{0.95\columnwidth}{
        \begin{tikzpicture}
            \node (s0b) [state, initial, inner sep=0pt, minimum size=1em] {};
            \node (s0) [state, below = of s0b, inner sep=0pt, minimum size=1em] {};
            \node (s1) [state, below = of s0, inner sep=0pt, minimum size=1em]       {};
            \node (s2) [state, below = of s1, inner sep=0pt, minimum size=1em]       {};
            \node (s3) [state, below = of s2, inner sep=0pt, minimum size=1em]       {};
            \node (s4) [state, below left = of s3, inner sep=0pt, minimum size=1em]       {};
            \node (s5) [state, below right = of s3, inner sep=0pt, minimum size=1em]       {};
            \node (s6) [state, below = of s4, inner sep=0pt, minimum size=1em]       {};
            \node (s7) [state, below = of s5, inner sep=0pt, minimum size=1em]       {};

            % Common part
            \path [->, thick] (s0b) edge node[right] {\lstinline{timer\_enable 3}/\,\mi{UMem\ 0}} (s0);
            \path [->, thick] (s0) edge node[right] {\lstinline{create\_enclave}/\,\mi{UMem\ 0}} (s1);
            \path [->, thick] (s1) edge node[right] {\lstinline{jin enclave\_start}/$\tau$} (s2);
            % Interference
            \path [->, thick, dashed, bend right, color=red, text=black] (s2) edge node[left] {\lstinline{cmp \#0, \#0}/$\tau$} (s3);
            \path [->, thick, dotted, bend left, color=blue, text=black] (s2) edge node[right] {\lstinline{cmp \#0, \#1}/$\tau$} (s3);
            \path [->, thick, dashed, bend right, color=red, text=black] (s3) edge node[left] {\lstinline{ifz_mov_nop}\,/$\tau$} (s4);
            \path [->, thick, dotted, bend left, color=blue, text=black] (s3) edge node[right] {\lstinline{ifz_mov_nop}\,/$\tau$} (s5);
            \path [->, thick, dashed, color=red, text=black] (s4) edge node[left] {\lstinline{IRQ}\,/\mi{Handle\, 4},\,\mi{UMem\ 42}} (s6);
            \path [->, thick, dotted, color=blue, text=black] (s5) edge node[right] {\lstinline{IRQ}\,/\mi{Handle\, 4},\,\mi{UMem\ 0}} (s7);
        \end{tikzpicture}
    }

    \caption{\small (Simplified) Witness graph for \textbf{V-B6} found by \toolname: the attacker observes differences in \mi{UMem}'s values induced by the victim, so deducing the secret value in the comparison.
    The action \lstinline|ifz (mov \#42, &unprot_mem; nop) (nop; mov \#42, &unprot_mem)| is shortened to \lstinline|ifz_mov_nop| to ease reading.}\label{fig:witness-b6}
\end{figure}
Notice that, this attack is not possible without interrupts since the attacker would only be able to observe value \lstinline|42| after the enclave returns.

\subsection{New Attacks}
As mentioned above, we found that $\SA{\Aadvanced}{T}{0} \not\oequiv \SA{\Aadvanced}{T}{1}$ even on the latest Sancus commit where all the mitigations to attacks of \cite{mindthegap} were implemented.
Since $\SA{\Abasic}{T}{0} \oequiv \SA{\Abasic}{T}{1}$ using the basic attacker, the non-equivalence between \SA{\Aadvanced}{T}{0} and \SA{\Aadvanced}{T}{1} corresponds to new attacks breaking~\cref{prop:fullinstance}.

The attacks we found are related to new undocumented discrepancies between \sancusv and its implementation, reported at the bottom of \cref{tab:mindthegap} in red color and described below.
For lack of space, witness graphs are reported in \cref{sec:appendix_witness}.

\subsubsection*{\textbf{V-B8} \vbeightdescr}
We found that read/write accesses to unprotected memory reset the CPU, while in \sancusv unauthorized accesses raise an exception that is handled explicitly.
To see how \toolname discovered this discrepancy, consider the following command of the enclave used in our tests (cf. \cref{fig:generalenclave}):
\begin{lstlisting}[numbers=none,frame=none,xleftmargin=10pt]
ifz (mov &unprot_mem, r8; nop) (nop; mov &unprot_mem, r8)
\end{lstlisting}
%
% Since reading from unprotected memory resets the CPU, if the attacker sets the timer right after the \lstinline|mov|, on the left branch the enclave will reset before the interrupt occurs \flaS{Mi sono persa qui e nell'es sotto, pensavo che il reset fosse dopo l'interrupt}  never reaching the \lstinline|isr| section and making it possible to distinguish the two branches and guess the secret guard. \ricS{Matteo, è corretto?}
Since reading from unprotected memory resets the CPU, if the attacker schedules an interrupt to arrive during the \lstinline|mov|, on the left branch the enclave will reset before the interrupt occurs  never reaching the \lstinline|isr| section and making it possible to distinguish the two branches and guess the secret guard (cf. \textbf{V-B8} in~\cref{sec:appendix_witness}).

We also found that when the enclave tries to jump to the data section, the CPU jumps to the location \lstinline|0x0240| and loops executing the instruction occurring at the address following that of the jump instruction.\footnote{Putting an instruction like \lstinline|add #1, r4| at such address allows us to see that register \lstinline|r4| keeps increasing while the CPU loops.}
Interestingly, this buggy behavior does not seem to be exploitable to leak the enclave secret.
In fact, consider the following enclave command:
\begin{lstlisting}[numbers=none,frame=none,xleftmargin=10pt]
ifz (jmp #$\datas$; nop) (nop; jmp #$\datas$)
\end{lstlisting}
By scheduling the interrupt as for the previous case, we can make the left branch loop before the interrupt occurs. However, differently from the reset, here the CPU is looping but handles the interrupts, so the attacker cannot distinguish between the two branches as they will both eventually execute the \lstinline|isr|.

\subsubsection*{\textbf{V-B9} \vbninedescr}
We found that the enclave can reset the CPU, while in \sancusv this is forbidden.
This is captured by the following enclave command:
\begin{lstlisting}[numbers=none,frame=none,xleftmargin=10pt]
ifz (rst; nop) (nop; rst)
\end{lstlisting}
As for read/write violations, the attacker can schedule an interrupt so that the left branch resets while the right one executes the \lstinline|isr| distinguishing the two cases (cf. \textbf{V-B9} in~\cref{sec:appendix_witness}).

% \item[\textbf{V-B10}] \vbtendescr
% \end{description}

% The attacks are due to two different behaviors that diverge from the \sancusv model: $(i)$ when the enclave executes forbidden commands such as reading and writing from/to the unprotected memory and manipulating the interrupt register, Sancus resets the CPU; $(ii)$ when the enclave tries to jump to the data section, the CPU jumps to unpredictable locations and loops. Resetting and looping might be hard to observe for the attacker, but the absence of reset or looping, instead, is clearly observable, an this makes a difference depending on whether the attacker can or cannot interrupt the enclave.

% A similar situation occurs for the following enclave command:
% \begin{lstlisting}[numbers=none,frame=none]
% ifz (jmp #$\datas$; nop) (nop; jmp #$\datas$)
% \end{lstlisting}
% With an appropriate timing for the interrupt, the CPU will loop in the left branch before the interrupt occurs, never reaching \lstinline|isr| and making the branches distinguishable.

\subsubsection*{Possible fixes}
Sancus developers have opted to reset the CPU when violations occur, but we have shown that this behavior is exploitable when enclaves are interruptible.
In \sancusv violations and resets inside the enclave raise exceptions that are handled explicitly.
Interestingly, this solution does not prevent the attack found by \toolname.
Instead, it makes code such as \lstinline|ifz (jmp #$\datas$; nop) (nop; jmp #$\datas$)| and \lstinline|ifz (rst; nop) (nop; rst)| vulnerable even under the basic attacker $\Abasic$.
Intuitively, exceptions make programs more vulnerable, but the effect is that noninterference is preserved.
This might appear counterintuitive but in fact exceptions are interruptions and, as such, make violations directly observable to the attacker.
A more conservative fix is to ignore violations when possible, e.g.,
%For example,
 read/write accesses to unprotected memory and explicit resets could be ignored preventing the attacks.
In this case, however, the implementation would diverge from \sancusv requiring a new formal proof of full abstraction.
% We are going to responsibly disclose our findings to Sancus developers and to discuss with them possible fixes.
We have
%responsibly
disclosed our findings to Sancus developers, and we are discussing with them possible fixes.

%-------------------------------------------------------------------------------
\section{Conclusion and future work}\label{sec:concl}
%-------------------------------------------------------------------------------
% !TEX root = ../main.tex

We presented a new method to \emph{bridge the gap} between models and implementations that first extracts a behavioral model by directly interacting with the real system, then precisely analyzes it to identify attacks and anomalies.
%The method
This is implemented in the \toolname tool~\cite{additionalmaterial}.
%, that we plan to make publicly available.
We used \toolname to systematically analyze Sancus,
%looking at
checking various commits and
%demonstrating
proving that the tool scales to a real, fully-functional secure architecture.
\toolname could be used to continuously checking  new commits, looking  for  possible regressions and new attacks, either through efficient tests that are automatically synthesized from previously-found attacks, or through a new round of learning and checking.

Note that, unlike testing, we do not simply look for attacks, but we formally verify security up to the correctness of the learned model which can precisely quantified  as explained in~\cref{sec:specification}.
% In our setting, this requires  first verifying  robust noninterference for the basic attacker, even before we can  look for issues with the advanced attacker.
In fact, our initial step involves verifying robust noninterference for the basic attacker before proceeding to identify potential issues with the advanced attacker.
Furthermore, it is well-known that noninterference is not a safety property (technically it is a 2-hypersafety~\cite{clarkson2010hyperproperties}) and it can only be refuted by exhibiting a pair of finite witness traces showing how attackers can influence the behavior of the victim in different ways.
While it is possible to extract these witnesses when noninterference is violated (\cref{fig:tree-toy}), the opposite does not hold: we cannot find a finite set of tests that imply noninterference.
Notice, in particular, that our systems can execute an unbounded number of actions, for example when scheduling a new interrupt in the \lstinline|isr| section.
In summary, analogously to what~\citet{chen2016pac} observed about correctness verification, with testing it is possible to find counterexamples to noninterference, but it is not possible to prove noninterference without learning a model of the system.

There are various extensions that we are planning as future work.
We intend to apply \toolname to other case studies.
The tool was designed in a modular way so that the learning and checking back-end (\cref{sec:learning,sec:checking}) can be fully reused for another architecture.
Specification (\cref{sec:specification}) instead is architecture-dependent and amounts to defining a suitable abstract behavior, the attackers and (possibly) the trusted code, but we oversee no critical issues with that.
Scalability is always an issue when analyzing real systems, in fact achieving strong security guarantees on larger/more complex architectures may require a lot of computational time (or may be infeasible).
% Interestingly,
\toolname makes it possible to overcome scalability issues at the price of a weaker security proof, e.g., by reducing the set of attackers.
%\matteoS{Rifrasato un po' qui per stressare la questione della scalability}
The main bottleneck for us was in the performance of Verilator-generated code that forced us to limit attacker/enclave specifications.
As a future work, we are planning to analyze a Sancus implementation in FPGA.

% %-------------------------------------------------------------------------------
% \section*{Acknowledgments}
% %-------------------------------------------------------------------------------

% \section*{Extra}
% %-------------------------------------------------------------------------------
% \input{sections/extra}

%-------------------------------------------------------------------------------
% \ric{Metto la biblio small altrimenti è enorme!}
\bibliographystyle{IEEEtranN}

{\small
\bibliography{references}

% Generated by IEEEtranN.bst, version: 1.14 (2015/08/26)
\begin{thebibliography}{47}
\providecommand{\natexlab}[1]{#1}
\providecommand{\url}[1]{#1}
\csname url@samestyle\endcsname
\providecommand{\newblock}{\relax}
\providecommand{\bibinfo}[2]{#2}
\providecommand{\BIBentrySTDinterwordspacing}{\spaceskip=0pt\relax}
\providecommand{\BIBentryALTinterwordstretchfactor}{4}
\providecommand{\BIBentryALTinterwordspacing}{\spaceskip=\fontdimen2\font plus
\BIBentryALTinterwordstretchfactor\fontdimen3\font minus \fontdimen4\font\relax}
\providecommand{\BIBforeignlanguage}[2]{{%
\expandafter\ifx\csname l@#1\endcsname\relax
\typeout{** WARNING: IEEEtranN.bst: No hyphenation pattern has been}%
\typeout{** loaded for the language `#1'. Using the pattern for}%
\typeout{** the default language instead.}%
\else
\language=\csname l@#1\endcsname
\fi
#2}}
\providecommand{\BIBdecl}{\relax}
\BIBdecl

\bibitem[Bognar et~al.(2022{\natexlab{a}})Bognar, Van~Bulck, and Piessens]{mindthegap}
M.~Bognar, J.~Van~Bulck, and F.~Piessens, ``{Mind the Gap: Studying the Insecurity of Provably Secure Embedded Trusted Execution Architectures},'' in \emph{2022 IEEE Symposium on Security and Privacy (SP22)}, 2022, pp. 1638--1655.

\bibitem[Herley and van Oorschot(2018)]{ScienceSecurity}
C.~Herley and P.~van Oorschot, ``{Science of Security: Combining Theory and Measurement to Reflect the Observable},'' \emph{IEEE Security \& Privacy}, vol.~16, no.~1, pp. 12--22, 2018.

\bibitem[{Goguen} and {Meseguer}(1982)]{goguen1982security}
J.~A. {Goguen} and J.~{Meseguer}, ``{Security Policies and Security Models},'' in \emph{IEEE Symposium on Security and Privacy}, 1982, pp. 11--20.

\bibitem[Patrignani et~al.(2019)Patrignani, Ahmed, and Clarke]{patrignani2019formal}
\BIBentryALTinterwordspacing
M.~Patrignani, A.~Ahmed, and D.~Clarke, ``{Formal Approaches to Secure Compilation: A Survey of Fully Abstract Compilation and Related Work},'' \emph{ACM Comput. Surv.}, vol.~51, no.~6, 2019. [Online]. Available: \url{https://doi.org/10.1145/3280984}
\BIBentrySTDinterwordspacing

\bibitem[Abate et~al.(2019)Abate, Blanco, Garg, Hritcu, Patrignani, and Thibault]{abate2018exploring}
C.~Abate, R.~Blanco, D.~Garg, C.~Hritcu, M.~Patrignani, and J.~Thibault, ``{Journey Beyond Full Abstraction: Exploring Robust Property Preservation for Secure Compilation},'' in \emph{32nd {IEEE} Computer Security Foundations Symposium (CSF19), Hoboken, NJ, USA, June 25-28, 2019}, 2019, pp. 256--271.

\bibitem[Busi et~al.(2021)Busi, Noorman, Bulck, Galletta, Degano, M{\"{u}}hlberg, and Piessens]{busi2021securing}
\BIBentryALTinterwordspacing
M.~Busi, J.~Noorman, J.~V. Bulck, L.~Galletta, P.~Degano, J.~T. M{\"{u}}hlberg, and F.~Piessens, ``{Securing Interruptible Enclaved Execution on Small Microprocessors},'' \emph{{ACM} Trans. Program. Lang. Syst.}, vol.~43, no.~3, pp. 12:1--12:77, 2021. [Online]. Available: \url{https://doi.org/10.1145/3470534}
\BIBentrySTDinterwordspacing

\bibitem[Focardi et~al.(2000)Focardi, Gorrieri, and Martinelli]{NIprotocols}
R.~Focardi, R.~Gorrieri, and F.~Martinelli, ``{Non Interference for the Analysis of Cryptographic Protocols},'' in \emph{Proceedings of the 27th International Colloquium on Automata, Languages and Programming}, ser. ICALP '00.\hskip 1em plus 0.5em minus 0.4em\relax Berlin, Heidelberg: Springer-Verlag, 2000, p. 354–372.

\bibitem[Durante et~al.(2000)Durante, Focardi, and Gorrieri]{CVS}
\BIBentryALTinterwordspacing
A.~Durante, R.~Focardi, and R.~Gorrieri, ``{A Compiler for Analyzing Cryptographic Protocols Using Noninterference},'' \emph{ACM Trans. Softw. Eng. Methodol.}, vol.~9, no.~4, p. 488–528, oct 2000. [Online]. Available: \url{https://doi.org/10.1145/363516.363532}
\BIBentrySTDinterwordspacing

\bibitem[Canella et~al.(2019)Canella, Bulck, Schwarz, Lipp, von Berg, Ortner, Piessens, Evtyushkin, and Gruss]{transient}
C.~Canella, J.~V. Bulck, M.~Schwarz, M.~Lipp, B.~von Berg, P.~Ortner, F.~Piessens, D.~Evtyushkin, and D.~Gruss, ``{A Systematic Evaluation of Transient Execution Attacks and Defenses},'' in \emph{28th {USENIX} Security Symposium, {USENIX} Security 2019}, 2019.

\bibitem[Herley and Van~Oorschot(2017)]{SoKattacks}
C.~Herley and P.~Van~Oorschot, ``{SoK: Science, Security and the Elusive Goal of Security as a Scientific Pursuit},'' in \emph{2017 IEEE Symposium on Security and Privacy (SP)}, 2017, pp. 99--120.

\bibitem[Jang et~al.(2010)Jang, Jhala, Lerner, and Shacham]{EmpiricalJavaScript}
\BIBentryALTinterwordspacing
D.~Jang, R.~Jhala, S.~Lerner, and H.~Shacham, ``{An Empirical Study of Privacy-Violating Information Flows in JavaScript Web Applications},'' in \emph{Proceedings of the 17th ACM Conference on Computer and Communications Security (CCS10)}.\hskip 1em plus 0.5em minus 0.4em\relax New York, NY, USA: Association for Computing Machinery, 2010, p. 270–283. [Online]. Available: \url{https://doi.org/10.1145/1866307.1866339}
\BIBentrySTDinterwordspacing

\bibitem[Basin and Capkun(2012)]{ResearchValueAttacks}
\BIBentryALTinterwordspacing
D.~Basin and S.~Capkun, ``{The Research Value of Publishing Attacks},'' \emph{Commun. ACM}, vol.~55, no.~11, p. 22–24, nov 2012. [Online]. Available: \url{https://doi.org/10.1145/2366316.2366324}
\BIBentrySTDinterwordspacing

\bibitem[Bourgeat et~al.(2020)Bourgeat, Pit-Claudel, Chlipala, and Arvind]{Koika}
\BIBentryALTinterwordspacing
T.~Bourgeat, C.~Pit-Claudel, A.~Chlipala, and Arvind, ``{The Essence of Bluespec: A Core Language for Rule-Based Hardware Design},'' in \emph{Proceedings of the 41st ACM SIGPLAN Conference on Programming Language Design and Implementation}, ser. PLDI 2020.\hskip 1em plus 0.5em minus 0.4em\relax New York, NY, USA: Association for Computing Machinery, 2020, p. 243–257. [Online]. Available: \url{https://doi.org/10.1145/3385412.3385965}
\BIBentrySTDinterwordspacing

\bibitem[Choi et~al.(2017)Choi, Vijayaraghavan, Sherman, Chlipala, and Arvind]{Kami}
\BIBentryALTinterwordspacing
J.~Choi, M.~Vijayaraghavan, B.~Sherman, A.~Chlipala, and Arvind, ``{Kami: A Platform for High-Level Parametric Hardware Specification and Its Modular Verification},'' \emph{Proc. ACM Program. Lang.}, vol.~1, no. ICFP, aug 2017. [Online]. Available: \url{https://doi.org/10.1145/3110268}
\BIBentrySTDinterwordspacing

\bibitem[Erbsen et~al.(2021)Erbsen, Gruetter, Choi, Wood, and Chlipala]{KamiCaseStudy}
\BIBentryALTinterwordspacing
A.~Erbsen, S.~Gruetter, J.~Choi, C.~Wood, and A.~Chlipala, ``{Integration Verification across Software and Hardware for a Simple Embedded System},'' in \emph{Proceedings of the 42nd ACM SIGPLAN International Conference on Programming Language Design and Implementation}, ser. PLDI 2021.\hskip 1em plus 0.5em minus 0.4em\relax New York, NY, USA: Association for Computing Machinery, 2021, p. 604–619. [Online]. Available: \url{https://doi.org/10.1145/3453483.3454065}
\BIBentrySTDinterwordspacing

\bibitem[Swamy et~al.(2013)Swamy, Chen, and Livshits]{Fstar1}
\BIBentryALTinterwordspacing
N.~Swamy, J.~Chen, and B.~Livshits, ``{Verifying Higher-order Programs with the Dijkstra Monad},'' in \emph{ACM Programming Language Design and Implementation (PLDI) 2013}.\hskip 1em plus 0.5em minus 0.4em\relax ACM, June 2013. [Online]. Available: \url{https://www.microsoft.com/en-us/research/publication/verifying-higher-order-programs-with-the-dijkstra-monad/}
\BIBentrySTDinterwordspacing

\bibitem[Swamy et~al.(2016)Swamy, Hritcu, Keller, Rastogi, Delignat-Lavaud, Forest, Bhargavan, Fournet, Strub, Kohlweiss, Zinzindohou\'e, and {Zanella-B\'eguelin}]{Fstar2}
\BIBentryALTinterwordspacing
N.~Swamy, C.~Hritcu, C.~Keller, A.~Rastogi, A.~Delignat-Lavaud, S.~Forest, K.~Bhargavan, C.~Fournet, P.-Y. Strub, M.~Kohlweiss, J.-K. Zinzindohou\'e, and S.~{Zanella-B\'eguelin}, ``Dependent types and multi-monadic effects in {F*},'' in \emph{43rd ACM SIGPLAN-SIGACT Symposium on Principles of Programming Languages (POPL)}.\hskip 1em plus 0.5em minus 0.4em\relax ACM, Jan. 2016, pp. 256--270. [Online]. Available: \url{https://www.fstar-lang.org/papers/mumon/}
\BIBentrySTDinterwordspacing

\bibitem[Chalupar et~al.(2014)Chalupar, Peherstorfer, Poll, and de~Ruiter]{reverseLego}
\BIBentryALTinterwordspacing
G.~Chalupar, S.~Peherstorfer, E.~Poll, and J.~de~Ruiter, ``Automated reverse engineering using {Lego{\textregistered}},'' in \emph{8th USENIX Workshop on Offensive Technologies (WOOT 14)}.\hskip 1em plus 0.5em minus 0.4em\relax San Diego, CA: USENIX Association, Aug. 2014. [Online]. Available: \url{https://www.usenix.org/conference/woot14/workshop-program/presentation/chalupar}
\BIBentrySTDinterwordspacing

\bibitem[Aarts et~al.(2013)Aarts, De~Ruiter, and Poll]{reverseBank}
F.~Aarts, J.~De~Ruiter, and E.~Poll, ``{Formal Models of Bank Cards for Free},'' in \emph{2013 IEEE Sixth International Conference on Software Testing, Verification and Validation Workshops}, 2013, pp. 461--468.

\bibitem[De~Ruiter and Poll(2015)]{UsenixTLSFuzzing}
J.~De~Ruiter and E.~Poll, ``{Protocol State Fuzzing of TLS Implementations},'' in \emph{Proceedings of the 24th USENIX Conference on Security Symposium}, ser. SEC'15.\hskip 1em plus 0.5em minus 0.4em\relax USA: USENIX Association, 2015, p. 193–206.

\bibitem[Pferscher and Aichernig(2021)]{BLEfingerprint}
A.~Pferscher and B.~K. Aichernig, ``{Fingerprinting Bluetooth Low Energy Devices via Active Automata Learning},'' in \emph{Formal Methods}, M.~Huisman, C.~P{\u{a}}s{\u{a}}reanu, and N.~Zhan, Eds.\hskip 1em plus 0.5em minus 0.4em\relax Cham: Springer International Publishing, 2021, pp. 524--542.

\bibitem[Nunes et~al.(2019)Nunes, Eldefrawy, Rattanavipanon, Steiner, and Tsudik]{vrased}
I.~O. Nunes, K.~Eldefrawy, N.~Rattanavipanon, M.~Steiner, and G.~Tsudik, ``{VRASED: A Verified Hardware/Software Co-Design for Remote Attestation},'' in \emph{28th {USENIX} Security Symposium, {USENIX} Security 2019}, 2019.

\bibitem[Irfan et~al.(2016)Irfan, Cimatti, Griggio, Roveri, and Sebastiani]{verilog2SMV}
A.~Irfan, A.~Cimatti, A.~Griggio, M.~Roveri, and R.~Sebastiani, ``{Verilog2SMV: A tool for word-level verification},'' in \emph{2016 Design, Automation \& Test in Europe Conference \& Exhibition (DATE)}, 2016, pp. 1156--1159.

\bibitem[Angluin(1987)]{angluin1987learning}
\BIBentryALTinterwordspacing
D.~Angluin, ``{Learning Regular Sets from Queries and Counterexamples},'' \emph{Inf. Comput.}, vol.~75, no.~2, pp. 87--106, 1987. [Online]. Available: \url{https://doi.org/10.1016/0890-5401(87)90052-6}
\BIBentrySTDinterwordspacing

\bibitem[Vaandrager(2017)]{vaandrager2017model}
\BIBentryALTinterwordspacing
F.~W. Vaandrager, ``{Model learning},'' \emph{Commun. {ACM}}, vol.~60, no.~2, pp. 86--95, 2017. [Online]. Available: \url{https://doi.org/10.1145/2967606}
\BIBentrySTDinterwordspacing

\bibitem[Bognar et~al.(2022{\natexlab{b}})Bognar, Van~Bulck, and Piessens]{mindthegap-repo}
\BIBentryALTinterwordspacing
M.~Bognar, J.~Van~Bulck, and F.~Piessens, ``{Online code repository for ``Mind the Gap: Studying the Insecurity of Provably Secure Embedded Trusted Execution Architectures''},'' 2022. [Online]. Available: \url{https://github.com/martonbognar/sancus-core-gap}
\BIBentrySTDinterwordspacing

\bibitem[Busi et~al.()Busi, Focardi, and Luccio]{additionalmaterial}
\BIBentryALTinterwordspacing
M.~Busi, R.~Focardi, and F.~Luccio, ``{ALVIE: tool, source code and experimental results},'' online; last access Feb 2024. [Online]. Available: \url{https://github.com/matteobusi/alvie}
\BIBentrySTDinterwordspacing

\bibitem[Distefano et~al.(2019)Distefano, F\"{a}hndrich, Logozzo, and O'Hearn]{facebookStatic}
\BIBentryALTinterwordspacing
D.~Distefano, M.~F\"{a}hndrich, F.~Logozzo, and P.~W. O'Hearn, ``Scaling static analyses at facebook,'' \emph{Commun. ACM}, vol.~62, no.~8, p. 62–70, jul 2019. [Online]. Available: \url{https://doi.org/10.1145/3338112}
\BIBentrySTDinterwordspacing

\bibitem[Rivest and Schapire(1993)]{rivest1993inference}
\BIBentryALTinterwordspacing
R.~L. Rivest and R.~E. Schapire, ``Inference of finite automata using homing sequences,'' \emph{Inf. Comput.}, vol. 103, no.~2, pp. 299--347, 1993. [Online]. Available: \url{https://doi.org/10.1006/inco.1993.1021}
\BIBentrySTDinterwordspacing

\bibitem[Vaandrager et~al.(2022)Vaandrager, Garhewal, Rot, and Wi{\ss}mann]{vaandrager2022new}
\BIBentryALTinterwordspacing
F.~W. Vaandrager, B.~Garhewal, J.~Rot, and T.~Wi{\ss}mann, ``{A New Approach for Active Automata Learning Based on Apartness},'' in \emph{Proceedings of the 28th International Conference on Tools and Algorithms for the Construction and Analysis of Systems (TACAS22), Munich, Germany, April 2-7, 2022}, ser. LNCS, D.~Fisman and G.~Rosu, Eds., vol. 13243.\hskip 1em plus 0.5em minus 0.4em\relax Springer, 2022, pp. 223--243. [Online]. Available: \url{https://doi.org/10.1007/978-3-030-99524-9\_12}
\BIBentrySTDinterwordspacing

\bibitem[Vila et~al.(2020)Vila, Ganty, Guarnieri, and K\"{o}pf]{CacheLearning}
\BIBentryALTinterwordspacing
P.~Vila, P.~Ganty, M.~Guarnieri, and B.~K\"{o}pf, ``{CacheQuery: Learning Replacement Policies from Hardware Caches},'' in \emph{Proceedings of the 41st ACM SIGPLAN Conference on Programming Language Design and Implementation}, ser. PLDI 2020.\hskip 1em plus 0.5em minus 0.4em\relax New York, NY, USA: Association for Computing Machinery, 2020, p. 519–532. [Online]. Available: \url{https://doi.org/10.1145/3385412.3386008}
\BIBentrySTDinterwordspacing

\bibitem[Fiter{\u{a}}u-Bro{\c{s}}tean et~al.(2016)Fiter{\u{a}}u-Bro{\c{s}}tean, Janssen, and Vaandrager]{TCPanalysis}
P.~Fiter{\u{a}}u-Bro{\c{s}}tean, R.~Janssen, and F.~Vaandrager, ``{Combining Model Learning and Model Checking to Analyze TCP Implementations},'' in \emph{Computer Aided Verification}, S.~Chaudhuri and A.~Farzan, Eds.\hskip 1em plus 0.5em minus 0.4em\relax Cham: Springer International Publishing, 2016, pp. 454--471.

\bibitem[Chen et~al.(2016)Chen, Hsieh, Leng{\'{a}}l, Lii, Tsai, Wang, and Wang]{chen2016pac}
\BIBentryALTinterwordspacing
Y.~Chen, C.~Hsieh, O.~Leng{\'{a}}l, T.~Lii, M.~Tsai, B.~Wang, and F.~Wang, ``{PAC} learning-based verification and model synthesis,'' in \emph{Proceedings of the 38th International Conference on Software Engineering, {ICSE} 2016, Austin, TX, USA, May 14-22, 2016}, L.~K. Dillon, W.~Visser, and L.~A. Williams, Eds.\hskip 1em plus 0.5em minus 0.4em\relax {ACM}, 2016, pp. 714--724. [Online]. Available: \url{https://doi.org/10.1145/2884781.2884860}
\BIBentrySTDinterwordspacing

\bibitem[Nemati et~al.(2020)Nemati, Buiras, Lindner, Guanciale, and Jacobs]{ValidationSide}
H.~Nemati, P.~Buiras, A.~Lindner, R.~Guanciale, and S.~Jacobs, ``{Validation of Abstract Side-Channel Models for Computer Architectures},'' in \emph{Computer Aided Verification}, S.~K. Lahiri and C.~Wang, Eds.\hskip 1em plus 0.5em minus 0.4em\relax Cham: Springer International Publishing, 2020, pp. 225--248.

\bibitem[Buiras et~al.(2021)Buiras, Nemati, Lindner, and Guanciale]{ValidationSide2}
\BIBentryALTinterwordspacing
P.~Buiras, H.~Nemati, A.~Lindner, and R.~Guanciale, ``{Validation of Side-Channel Models via Observation Refinement},'' in \emph{MICRO-54: 54th Annual IEEE/ACM International Symposium on Microarchitecture}, ser. MICRO '21.\hskip 1em plus 0.5em minus 0.4em\relax New York, NY, USA: Association for Computing Machinery, 2021, p. 578–591. [Online]. Available: \url{https://doi.org/10.1145/3466752.3480130}
\BIBentrySTDinterwordspacing

\bibitem[Sutton et~al.(2007)Sutton, Greene, and Amini]{FuzzingBook}
M.~Sutton, A.~Greene, and P.~Amini, \emph{{Fuzzing: brute force vulnerability discovery}}.\hskip 1em plus 0.5em minus 0.4em\relax Pearson Education, 2007.

\bibitem[Trippel et~al.(2022)Trippel, Shin, Chernyakhovsky, Kelly, Rizzo, and Hicks]{FuzzingHardware}
\BIBentryALTinterwordspacing
T.~Trippel, K.~G. Shin, A.~Chernyakhovsky, G.~Kelly, D.~Rizzo, and M.~Hicks, ``{Fuzzing Hardware Like Software},'' in \emph{31st USENIX Security Symposium (USENIX Security 22)}.\hskip 1em plus 0.5em minus 0.4em\relax Boston, MA: USENIX Association, Aug. 2022, pp. 3237--3254. [Online]. Available: \url{https://www.usenix.org/conference/usenixsecurity22/presentation/trippel}
\BIBentrySTDinterwordspacing

\bibitem[Aschermann et~al.(2019)Aschermann, Frassetto, Holz, Jauernig, Sadeghi, and Teuchert]{FuzzingGrammar}
\BIBentryALTinterwordspacing
C.~Aschermann, T.~Frassetto, T.~Holz, P.~Jauernig, A.~Sadeghi, and D.~Teuchert, ``{NAUTILUS:} fishing for deep bugs with grammars,'' in \emph{26th Annual Network and Distributed System Security Symposium, {NDSS} 2019, San Diego, California, USA, February 24-27, 2019}.\hskip 1em plus 0.5em minus 0.4em\relax The Internet Society, 2019. [Online]. Available: \url{https://www.ndss-symposium.org/ndss-paper/nautilus-fishing-for-deep-bugs-with-grammars/}
\BIBentrySTDinterwordspacing

\bibitem[Clarkson and Schneider(2010)]{clarkson2010hyperproperties}
M.~R. Clarkson and F.~B. Schneider, ``{Hyperproperties},'' \emph{Journal of Computer Security}, vol.~18, no.~6, pp. 1157--1210, 2010.

\bibitem[Noorman et~al.(2017{\natexlab{a}})Noorman, Bulck, M{\"{u}}hlberg, Piessens, Maene, Preneel, Verbauwhede, G{\"{o}}tzfried, M{\"{u}}ller, and Freiling]{noorman2017sancus}
J.~Noorman, J.~V. Bulck, J.~T. M{\"{u}}hlberg, F.~Piessens, P.~Maene, B.~Preneel, I.~Verbauwhede, J.~G{\"{o}}tzfried, T.~M{\"{u}}ller, and F.~C. Freiling, ``Sancus 2.0: {A} low-cost security architecture for iot devices,'' \emph{{ACM} Trans. Priv. Secur.}, vol.~20, no.~3, pp. 7:1--7:33, 2017.

\bibitem[Muskardin et~al.(2022)Muskardin, Aichernig, Pill, Pferscher, and Tappler]{muskardin2022aalpy}
E.~Muskardin, B.~K. Aichernig, I.~Pill, A.~Pferscher, and M.~Tappler, ``Aalpy: an active automata learning library,'' \emph{Innov. Syst. Softw. Eng.}, vol.~18, no.~3, pp. 417--426, 2022.

\bibitem[{Verilator Team}()]{verilator}
{Verilator Team}, ``{Verilator: Open-source SystemVerilog Simulator and Lint System},'' \url{https://verilator.org/}, online; last access May 2023.

\bibitem[Noorman et~al.(2017{\natexlab{b}})Noorman, Bulck, M\"{u}hlberg, Piessens, Maene, Preneel, Verbauwhede, G\"{o}tzfried, M\"{u}ller, and Freiling]{sancus}
\BIBentryALTinterwordspacing
J.~Noorman, J.~V. Bulck, J.~T. M\"{u}hlberg, F.~Piessens, P.~Maene, B.~Preneel, I.~Verbauwhede, J.~G\"{o}tzfried, T.~M\"{u}ller, and F.~Freiling, ``{Sancus 2.0: A Low-Cost Security Architecture for IoT Devices},'' \emph{ACM Trans. Priv. Secur.}, vol.~20, no.~3, pp. 7:1--7:33, Jul. 2017. [Online]. Available: \url{http://doi.acm.org/10.1145/3079763}
\BIBentrySTDinterwordspacing

\bibitem[McKeen et~al.(2013)McKeen, Alexandrovich, Berenzon, Rozas, Shafi, Shanbhogue, and Savagaonkar]{sgx}
F.~McKeen, I.~Alexandrovich, A.~Berenzon, C.~V. Rozas, H.~Shafi, V.~Shanbhogue, and U.~R. Savagaonkar, ``{Innovative instructions and software model for isolated execution},'' in \emph{{HASP} 2013, The Second Workshop on Hardware and Architectural Support for Security and Privacy, Tel-Aviv, Israel, June 23-24, 2013}, R.~B. Lee and W.~Shi, Eds.\hskip 1em plus 0.5em minus 0.4em\relax {ACM}, 2013, p.~10.

\bibitem[Instruments()]{ti-msp430}
T.~Instruments, ``{MSP430x1xx Family: User Guide},'' \url{http://www.ti.com/lit/ug/slau049f/slau049f.pdf}.

\bibitem[Girard()]{openMSP430}
\BIBentryALTinterwordspacing
O.~Girard, ``{OpenCores: OpenMSP430},'' online; last access May 2023. [Online]. Available: \url{https://opencores.org/projects/openmsp430}
\BIBentrySTDinterwordspacing

\bibitem[Van~Bulck et~al.()Van~Bulck, Piessens, and Strackx]{nemesis}
\BIBentryALTinterwordspacing
J.~Van~Bulck, F.~Piessens, and R.~Strackx, ``Nemesis: Studying microarchitectural timing leaks in rudimentary {CPU} interrupt logic,'' in \emph{Proceedings of the 2018 ACM SIGSAC Conference on Computer and Communications Security}, ser. CCS '18.\hskip 1em plus 0.5em minus 0.4em\relax ACM, pp. 178--195. [Online]. Available: \url{http://doi.acm.org/10.1145/3243734.3243822}
\BIBentrySTDinterwordspacing

\end{thebibliography}
}

\FloatBarrier
% \newpage
% \appendix
\appendices % This is unique to IEEE...
\crefalias{section}{appendix}

% \section{Synthesizing tests from witnesses}\label{sec:appendix_synthesis}
%   \input{sections/appendix_synthesis}

% \section{Commit history}\label{sec:appendix_history}
%   \input{sections/appendix_commits}

\section{Witness graphs}\label{sec:appendix_witness}
  % !TEX root = ../main.tex
\subsection{Witness graph for \textbf{V-B1}}
% \label{sec:WB1}
%
Recall that the~\textbf{V-B1} implementation-model mismatch states that the first instruction after a \lstinline|reti| takes one cycle more than expected.
\cref{fig:graph-b1} reports the witness graph that \toolname produces for~\textbf{V-B1} on the original commit (\origcommit) of~\cite{mindthegap-repo}.
The attack re-discovered by \toolname works as follows: first the attacker uses \lstinline|timer_enable 3| to schedule an interrupt during the first instruction of the \lstinline|ifz (add #1, &data_s; nop) (nop; add #1, &data_s)| action.
Then, when the interrupt handler is invoked, it schedules another interrupt for the first cycle after exiting the handler (\lstinline|timer_enable 1|).
Notice that, when invoking \lstinline|timer_enable k|, the Timer A device of Sancus starts counting between $0$ and $k$, so in our case the observed value for $\mi{TimerA}$ is either $0$ or $1$.
So, when the interrupt handler returns (\lstinline|reti|):
\begin{itemize}
    \item If \lstinline|s| equals 0, then the first CPU cycle after \lstinline|reti| is part of the padding put in place by the processor to mitigate the Nemesis attack~\cite{nemesis}.
    When the padding completes the control is handed over directly to the interrupt handler and the attacker observes $\mi{TimerA}\ 0$.
    \item Otherwise, the control goes back to the enclave and the CPU executes the remaining \lstinline|nop| instruction from \lstinline|ifz (add #1, &data_s; nop) (nop; add #1, &data_s)| before invoking the interrupt handler.
    Due to the bug, \lstinline|nop| takes $2$ cycles instead of $1$ and the \lstinline|ifz| branch is not balanced anymore.
    When the attacker re-gains control it observes $\mi{TimerA}\ 1$.
\end{itemize}
\begin{figure}[tb]
    \centering
    \resizebox{0.85\columnwidth}{!}
    {
        \begin{tikzpicture}
            \node (s0b) [state, initial, inner sep=0pt, minimum size=1em] {};
            \node (s0) [state, below = of s0b, inner sep=0pt, minimum size=1em] {};
            \node (s1) [state, below = of s0, inner sep=0pt, minimum size=1em]       {};
            \node (s2) [state, below = of s1, inner sep=0pt, minimum size=1em]       {};
            \node (s3) [state, below = of s2, inner sep=0pt, minimum size=1em]       {};
            \node (s4) [state, below = of s3, inner sep=0pt, minimum size=1em]       {};
            \node (s5) [state, below = of s4, inner sep=0pt, minimum size=1em]       {};
            \node (s6) [state, below = of s5, inner sep=0pt, minimum size=1em]       {};

            \node (s7) [state, below left = of s6, inner sep=0pt, minimum size=1em]       {};
            \node (s8) [state, below right = of s6, inner sep=0pt, minimum size=1em]       {};

            \node (s9) [state, below = of s8, inner sep=0pt, minimum size=1em]       {};

            % Common part
            \path [->, thick] (s0b) edge node[right] {\lstinline{timer\_enable 3}/\,\mi{TimerA\ 0}} (s0);
            \path [->, thick] (s0) edge node[right] {\lstinline{create\_enclave}/\,\mi{TimerA\ 0}} (s1);
            \path [->, thick] (s1) edge node[right] {\lstinline{jin enclave\_start}/$\tau$} (s2);
            % Interference
            \path [->, thick, dashed, bend right, color=red, text=black] (s2) edge node[left] {\lstinline{cmp \#0, \#0}/$\tau$} (s3);
            \path [->, thick, dotted, bend left, color=blue, text=black] (s2) edge node[right] {\lstinline{cmp \#0, \#1}/$\tau$} (s3);

            \path [->, thick] (s3) edge node[left]  {\lstinline{ifz_add_nop}\,/$\tau$} (s4);
            \path [->, thick] (s4) edge node[right] {\lstinline{timer\_enable 1}/\,\mi{TimerA\ 0}} (s5);
            \path [->, thick] (s5) edge node[right] {\lstinline{reti}/\,\mi{Reti}} (s6);

            \path [->, thick, dashed, bend right, color=red, text=black] (s6) edge node[left] {\underline{\lstinline{IRQ}\,/$\mi{Handle}\ 9,\,\mi{TimerA}\, 1$}} (s7);

            \path [->, thick, dotted, bend left, color=blue, text=black] (s6) edge node[right] {\lstinline{rem\_nop}\,/$\tau$} (s8);

            \path[->, thick, dotted, color=blue, text=black] (s8) edge node[right] {\underline{\lstinline{IRQ}\,/$\mi{Handle}\ 7,\,\mi{TimerA}\, 0$}} (s9);
            % %
            % \path [->, thick, dotted, color=blue, text=black] (s5) edge node[right] {\underline{\lstinline{IRQ}\,/\mi{Handle\, 4},\,\mi{UMem\ 0}}} (s7);
        \end{tikzpicture}
    }

    \caption{\small (Simplified) Witness graph for \textbf{V-B1} found by \toolname: the attacker observes differences in behavior (underlined), so deducing the secret value in the comparison.
    Here, \lstinline|ifz (add \#1, &data_s; nop) (nop; add \#1, &data_s)| is shortened to \lstinline|ifz_add_nop| to ease reading, and \lstinline|rem_nop| denotes the \lstinline|nop| instruction of the \lstinline|ifz_add_nop| left to be executed after returning from the interrupt.}\label{fig:graph-b1}
\end{figure}
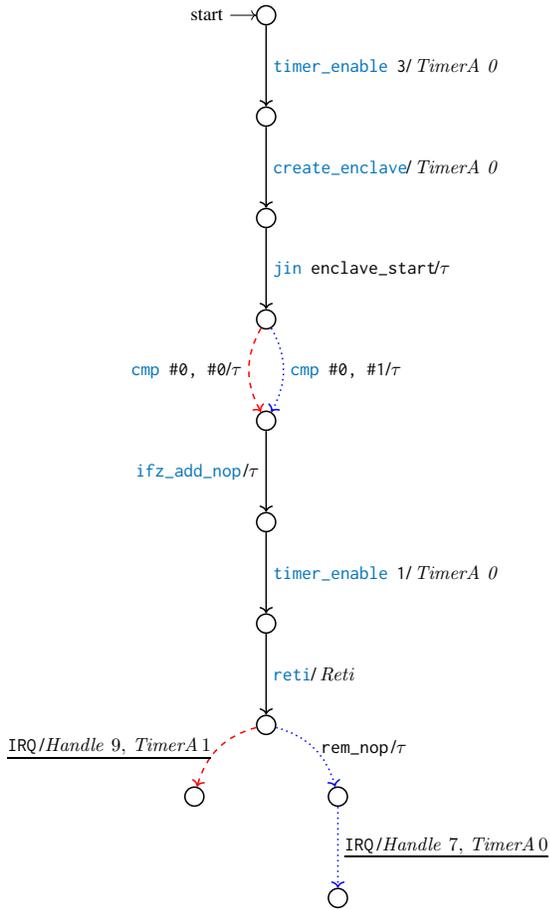

\subsection{Witness graph for \textbf{V-B8}}
% \label{sec:WB8}
%
Recall that the novel implementation-model~\textbf{V-B8} observes that read/write violations reset the CPU.
\cref{fig:graph-b8} reports the witness graph that \toolname produces for~\textbf{V-B8} on the last commit (\lastcommit) of~\cite{mindthegap-repo} and supports the explanation given in~\cref{sec:sancus}.
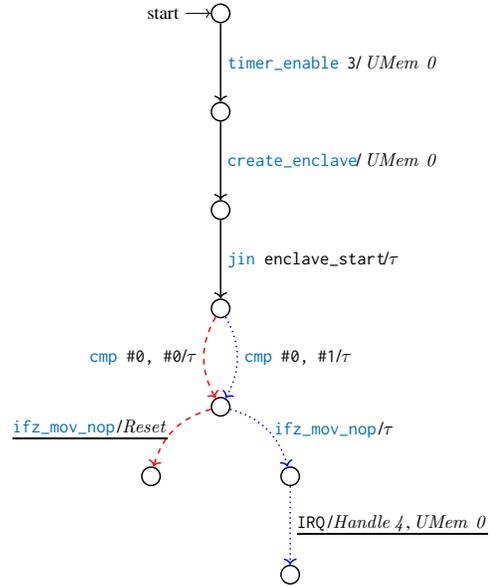
\begin{figure}[tb]
    \centering
    \resizebox{0.75\columnwidth}{!}{
        \begin{tikzpicture}
            \node (s0b) [state, initial, inner sep=0pt, minimum size=1em] {};
            \node (s0) [state, below = of s0b, inner sep=0pt, minimum size=1em] {};
            \node (s1) [state, below = of s0, inner sep=0pt, minimum size=1em]       {};
            \node (s2) [state, below = of s1, inner sep=0pt, minimum size=1em]       {};
            \node (s3) [state, below = of s2, inner sep=0pt, minimum size=1em]       {};
            \node (s4) [state, below left = of s3, inner sep=0pt, minimum size=1em]       {};
            \node (s5) [state, below right = of s3, inner sep=0pt, minimum size=1em]       {};
            \node (s7) [state, below = of s5, inner sep=0pt, minimum size=1em]       {};

            % Common part
            \path [->, thick] (s0b) edge node[right] {\lstinline{timer\_enable 3}/\,\mi{UMem\ 0}} (s0);
            \path [->, thick] (s0) edge node[right] {\lstinline{create\_enclave}/\,\mi{UMem\ 0}} (s1);
            \path [->, thick] (s1) edge node[right] {\lstinline{jin enclave\_start}/$\tau$} (s2);
            % Interference
            \path [->, thick, dashed, bend right, color=red, text=black] (s2) edge node[left] {\lstinline{cmp \#0, \#0}/$\tau$} (s3);
            \path [->, thick, dotted, bend left, color=blue, text=black] (s2) edge node[right] {\lstinline{cmp \#0, \#1}/$\tau$} (s3);
            \path [->, thick, dashed, bend right, color=red, text=black] (s3) edge node[left] {\underline{\lstinline{ifz_mov_nop}\,/\mi{Reset}}} (s4);
            \path [->, thick, dotted, bend left, color=blue, text=black] (s3) edge node[right] {\lstinline{ifz_mov_nop}\,/$\tau$} (s5);
            \path [->, thick, dotted, color=blue, text=black] (s5) edge node[right] {\underline{\lstinline{IRQ}\,/\mi{Handle\, 4},\,\mi{UMem\ 0}}} (s7);
        \end{tikzpicture}
    }

    \caption{\small (Simplified) Witness graph for \textbf{V-B8} found by \toolname: the attacker observes differences in behavior (\mi{Reset} vs. \mi{Handle}, both underlined), so deducing the secret value in the comparison.
    The action \lstinline|ifz (mov \&unprot\_mem, r8; nop) (nop; mov \&unprot\_mem, r8)| is shortened to \lstinline|ifz_mov_nop| to ease reading.}\label{fig:graph-b8}
\end{figure}

\subsection{Witness graph for \textbf{V-B9}}
% \label{sec:WB9}
%
Recall that the novel implementation-model~\textbf{V-B9} observes that enclaves can reset the CPU explicitly.
\cref{fig:graph-b9} reports the witness graph that \toolname produces for~\textbf{V-B9} on the last commit (\lastcommit) of~\cite{mindthegap-repo} and supports the explanation given in~\cref{sec:sancus}.
\begin{figure}[t]
    \centering
    \resizebox{0.7\columnwidth}{!}{
        \begin{tikzpicture}
            \node (s0b) [state, initial, inner sep=0pt, minimum size=1em] {};
            \node (s0) [state, below = of s0b, inner sep=0pt, minimum size=1em] {};
            \node (s1) [state, below = of s0, inner sep=0pt, minimum size=1em]       {};
            \node (s2) [state, below = of s1, inner sep=0pt, minimum size=1em]       {};
            \node (s3) [state, below = of s2, inner sep=0pt, minimum size=1em]       {};
            \node (s4) [state, below left = of s3, inner sep=0pt, minimum size=1em]       {};
            \node (s5) [state, below right = of s3, inner sep=0pt, minimum size=1em]       {};
            %\node (s6) [state, below = of s4, inner sep=0pt, minimum size=1em]       {};
            \node (s7) [state, below = of s5, inner sep=0pt, minimum size=1em]       {};

            % Common part
            \path [->, thick] (s0b) edge node[right] {\lstinline{timer\_enable 3}/\,\mi{UMem\ 0}} (s0);
            \path [->, thick] (s0) edge node[right] {\lstinline{create\_enclave}/\,\mi{UMem\ 0}} (s1);
            \path [->, thick] (s1) edge node[right] {\lstinline{jin enclave\_start}/$\tau$} (s2);
            % Interference
            \path [->, thick, dashed, bend right, color=red, text=black] (s2) edge node[left] {\lstinline{cmp \#0, \#0}/$\tau$} (s3);
            \path [->, thick, dotted, bend left, color=blue, text=black] (s2) edge node[right] {\lstinline{cmp \#0, \#1}/$\tau$} (s3);
            \path [->, thick, dashed, bend right, color=red, text=black] (s3) edge node[left] {\lstinline{ifz_rst_nop}\,/$\mi{Reset}$} (s4);
            \path [->, thick, dotted, bend left, color=blue, text=black] (s3) edge node[right] {\lstinline{ifz_rst_nop}\,/$\tau$} (s5);
            \path [->, thick, dotted, color=blue, text=black] (s5) edge node[right] {\lstinline{IRQ}\,/\mi{Handle\, 9}} (s7);
        \end{tikzpicture}
    }

    \caption{\small (Simplified) Witness graph for \textbf{V-B9} found by \toolname: the attacker observes differences in behavior (victim's \mi{Reset} vs. \mi{Handle}, both underlined), so deducing the secret value in the comparison.
    The action \lstinline|ifz (rst; nop) (nop; rst)| is shortened to \lstinline|ifz_rst_nop| to ease reading.}\label{fig:graph-b9}
\end{figure}
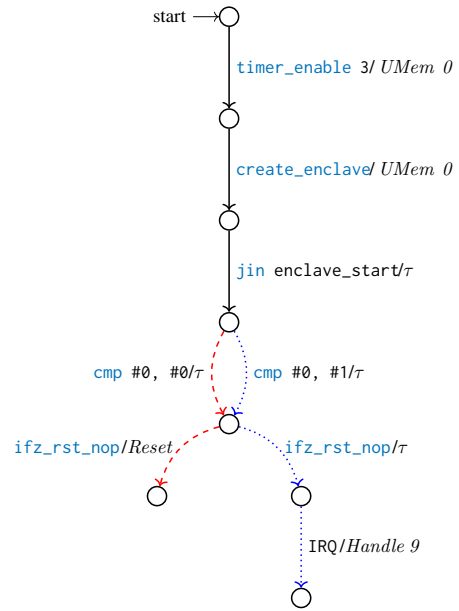

%%%%%%%%%%%%%%%%%%%%%%%%%%%%%%%%%%%%%%%%%%%%%%%%%%%%%%%%%%%%%%%%%%%%%%%%%%%%%%%%
% \clearpage
\end{document}